\documentclass[journal,10pt]{IEEEtran}
\usepackage{amssymb}
\usepackage{amsmath}
\usepackage{cite}
\usepackage{url}
\usepackage{xcolor}
\usepackage{cite,graphicx,amsmath,amssymb}
\usepackage{subfigure}
\usepackage{fancyhdr}
\usepackage{mdwmath}
\usepackage{mdwtab}
\usepackage{caption}
\usepackage{amsthm}
\usepackage{setspace}
\usepackage{algorithm}
\usepackage{algorithmic}
\usepackage{makecell}
\usepackage{diagbox}
\usepackage{stfloats}
\usepackage{balance}

\newtheorem{theorem}{Theorem}

\newtheorem{lemma}{Lemma}

\newtheorem{corollary}{Corollary}

\allowdisplaybreaks
\makeatletter
\def\ScaleIfNeeded{
	\ifdim\Gin@nat@width>\linewidth \linewidth \else \Gin@nat@width
	\fi } \makeatother
 
\begin{document}
\title{Robust Resource Allocation for STAR-RIS Assisted SWIPT Systems}

\author{
	Guangyu~Zhu,
	Xidong~Mu,~\IEEEmembership{Member,~IEEE,}
	Li~Guo,~\IEEEmembership{Member,~IEEE,}
	Ao~Huang,
	and Shibiao~Xu,~\IEEEmembership{Member,~IEEE}
		
\thanks{Part of this work was presented at the IEEE International Conference on Communications (ICC), Rome, Italy, May 28– June 1, 2023.\cite{Zhu_conference}}
\thanks{Guangyu Zhu, Li Guo, Ao Huang and Shibiao Xu are with the Key Laboratory of Universal Wireless Communications, Ministry of Education, Beijing University of Posts and Telecommunications, Beijing 100876, China, also with the School of Artificial Intelligence, Beijing University of Posts and Telecommunications, Beijing 100876, China, also with the Engineering Research Center of Ministry of Education for Chain Network Convergence Technology, and also with the National Engineering Research Center for Mobile Internet Security Technology, Beijing University of Posts and Telecommunications, Beijing 100876, China (email:\{Zhugy, guoli, huangao, shibiaoxu\}@bupt.edu.cn).}
\thanks{Xidong Mu is with the School of Electronic Engineering and Computer Science, Queen Mary University of London, London E1 4NS, U.K. (e-mail:\{xidong.mu\}@qmul.ac.uk).}
}
	
\maketitle
\begin{abstract}
A simultaneously transmitting and reflecting reconfigurable intelligent surface (STAR-RIS) assisted simultaneous wireless information and power transfer (SWIPT) system is proposed. More particularly, an STAR-RIS is deployed to assist in the information/power transfer from a multi-antenna access point (AP) to multiple single-antenna information users (IUs) and energy users (EUs), where two practical STAR-RIS operating protocols, namely energy splitting (ES) and time switching (TS), are employed. Under the imperfect channel state information (CSI) condition, a multi-objective optimization problem (MOOP) framework, that simultaneously maximizes the minimum data rate and minimum harvested power, is employed to investigate the fundamental rate-energy trade-off between IUs and EUs. To obtain the optimal robust resource allocation strategy, the MOOP is first transformed into a single-objective optimization problem (SOOP) via the $\epsilon$-constraint method, which is then reformulated by approximating semi-infinite inequality constraints with the S-procedure. For ES, an alternating optimization (AO)-based algorithm is proposed to jointly design AP active beamforming and STAR-RIS passive beamforming, where a penalty method is leveraged in STAR-RIS beamforming design. Furthermore, the developed algorithm is extended to optimize the time allocation policy and beamforming vectors in a two-layer iterative manner for TS. Numerical results reveal that: 1) deploying STAR-RISs achieves a significant performance gain over conventional RISs, especially in terms of harvested power for EUs; 2) the ES protocol obtains a better user fairness performance when focusing only on IUs or EUs, while the TS protocol yields a better balance between IUs and EUs; 3) the imperfect CSI affects IUs more significantly than EUs, whereas TS can confer a more robust design to attenuate these effects.
\end{abstract}
\begin{IEEEkeywords}
	Reconfigurable intelligent surfaces, simultaneous transmission and reflection, simultaneous wireless information and power transfer, resource allocation, imperfect CSI.
\end{IEEEkeywords}
\section{Introduction}
With the rapid growth in the number of Internet-of-Things (IoT) devices, the sustainability of equipment and energy supply has become one of the bottlenecks restricting the development of next generation wireless networks \cite{Zhang_energy,Buzzi_energy}. To address this issue, simultaneous wireless information and power transfer (SWIPT) as a promising technique has been investigated in \cite{SWIPT,Zeng_SWIPT}. In the SWIPT system, as the carrier of both information and energy, radio-frequency (RF) signals can realize the parallel information exchange and energy supply for IoT devices. In this case, balancing wireless communication and energy transfer becomes an important criterion in the SWIPT system \cite{ZR1}. However, due to the drastically different power sensitivities to receiver and transmission environments, wireless power transfer (WPT) has a significant efficiency gap compared to wireless information transfer (WIT). Especially over a long communication distance, the propagation loss will seriously reduce the efficiency of WPT, which poses a huge challenge for resource allocation to balance the performance of information users (IUs) and energy users (EUs). In this connection, smart antenna technologies are introduced into the SWIPT system to improve user performance \cite{Ding_MIMO}. At the same time, however, it should be noted that the increase in active antennas will be accompanied by more energy consumption, higher design complexity, and higher hardware costs. This virtually brings new obstacles to implementation in practice. 
	
Recently, reconfigurable intelligent surfaces (RISs)\cite{RIS_survey} have emerged as a key rising technology for future communication networks, thanks to their impressive performance in terms of energy efficiency \cite{Huang_EE}. An RIS is a two-dimensional (2D) surface, comprising large numbers of low-cost and passive meta-materials with tunable reflection properties. Ideally, each reconfigurable element can independently adjust the phase shift and amplitude of the incident signals according to the practical transmission needs, thus facilitating the creation of “Smart Radio Environments (SREs)” \cite{SRE}. Besides, due to the near-passive operation and low hardware footprints, RIS assisted system is more cost-efficient and  deployment-flexible compared to conventional active antenna systems.
Given these appealing features, the combination of RISs and SWIPT has attracted increasing attention \cite{SWIPT_Application}.
	
However, limited by their reflection-only nature, RISs can only serve transmitters and receivers located in the same side. This strict geographical restriction severely undermines user fairness and the effectiveness of introducing RISs. To break this limitation and thus achieve \emph{full}-space SREs, a novel concept of RISs, i.e., simultaneously transmitting and reflecting RISs (STAR-RISs), has been proposed \cite{Mu_star}. Unlike conventional reflecting-only RISs, STAR-RISs can not only reflect but also transmit the incident signals via dynamic configuration adjustment for all users located in both sides of STAR-RISs \cite{Liu_360}. Furthermore, STAR-RISs can exploit extra degrees of freedom (DoFs) for system design \cite{STAR-RIS_XU}, which may enable more flexible wireless resource allocation to coordinate the performance between WPT and WIT for SWIPT systems.
	
\subsection{Prior Works}
1) \emph{Studies on Resource Allocation for RIS Assisted SWIPT}: Enlightened by the enormous potential gains of RISs to transmission efficiency, the new research paradigm of resource allocation for RIS assisted SWIPT has been extensively studied recently in \cite{Wu_Weighted,Pan_MIMO,Wu_QoS,Zargari2021joint,Xu_NL,Zargari_Robust,Multi_objective_IRS,Waveform,XU_large}. In particular, the authors of \cite{Wu_Weighted} proposed a single-objective optimization problem (SOOP) for maximizing the weighted sum-power, subject to the individual signal-to-interference-plus-noise ratio (SINR) constraints for IUs in RIS aided multiple input single output (MISO) SWIPT systems. On the contrary, the maximization problem for weighted sum-rate, subject to the total harvested power constraints for EUs in RIS assisted multiple input multiple output (MIMO) SWIPT systems was considered in \cite{Pan_MIMO}. In other aspects, under the quality-of-service (QoS) requirements, the authors of \cite{Wu_QoS} proposed to optimize the joint active and passive beamforming for transmit power minimization in a multiple RISs assisted SWIPT system. The authors further proposed two low-complexity suboptimal algorithms based on maximum ratio transmission (MRT) and zero-forcing (ZF) beamforming techniques to minimize the transmit power at the BS in \cite{Zargari2021joint}. Besides, the authors of \cite{Xu_NL} investigated a similar transmit power minimization problem for a large-scale RIS aided SWIPT system.
Considering further the non-linear harvested model, the authors of \cite{Zargari_Robust} employed a block coordinate descent (BCD) method to jointly optimize the active beamforming and passive beamforming for minimization of transmit power under the imperfect channel state information (CSI) assumption.
However, the above studies only focus on either data rate or harvested power. To better portray the conflict between IUs and EUs in SWIPT, the rate-energy trade-offs were investigated in\cite{Multi_objective_IRS,Waveform,XU_large}. Specifically, a multi-objective optimization problem (MOOP) framework was studied by the authors of \cite{Multi_objective_IRS}, where energy/information beamforming vectors at the BS and phase shifts at the RIS were jointly optimized to obtain the fundamental trade-off between sum-rate and total harvested energy. In \cite{Waveform}, the authors jointly optimized waveform, active and passive beamforming by applying a BCD method to achieve a rate-energy region. Besides, the authors of \cite{XU_large} adopted a realistic-based RIS model in a RIS assisted MISO multiuser SWIPT system, where the trade-off between IUs and EUs was investigated by a penalty-based algorithm. 
	
2) \emph{Studies on STAR-RIS Assisted Communication}: There have been some early efforts to capitalize on the benefits of STAR-RIS deployment in wireless communication systems \cite{Mu_star,Wu_resource,Zhao_UAV,MEC,STAR_SWIPT}. In \cite{Mu_star}, the authors studied a STAR-RIS assisted two-user downlink MISO system and proposed three different protocols, i.e., energy splitting (ES), mode switching (MS), and time switching (TS), for practical operation. To explore the advantages of each protocol, the authors further study the power consumption minimization problems of unicast and multicast on different protocols. In \cite{Wu_resource}, the authors considered a STAR-RIS aided multi-carrier communication network, where the resource allocation design was achieved for both non-orthogonal multiple access (NOMA) and orthogonal multiple access (OMA) via jointly optimizing the channel assignment, power allocation, and the STAR-RIS configuration. The authors of \cite{Zhao_UAV} studied a STAR-RIS aided unmanned aerial vehicle (UAV) system and proposed a novel distributionally-robust reinforcement learning algorithm to get a robust design for sum-rate maximization. In \cite{MEC}, resource allocation for STAR-RIS assisted wireless powered mobile edge computing (MEC) systems was investigated, which illustrated that TS is a more suitable protocol for uplink transmission in MEC. Besides, the STAR-RIS assisted SWIPT was introduced in \cite{STAR_SWIPT}, where the achievable worst-case sum secrecy rate was maximized by applying a convex approximation method.
	
\subsection{Motivations and Contributions}
Although the introduction of RISs brings significant performance improvements to SWIPT systems, the \emph{half}-space coverage feature of reflecting-only RISs greatly limits its applicability in practical implementation. Benefiting from the simultaneous transmission and reflection characteristics, the STAR-RIS is envisaged as a potential technology to break this fundamental limitation. By leveraging enhanced DoFs, the STAR-RIS can not only improve SWIPT system performance, but also realize \emph{full}-space coverage for both IUs and EUs. Inspired by this, we propose to introduce the STAR-RIS into the SWIPT system in this paper. To the best of our knowledge, few efforts have been devoted to this topic. Therefore, there are several urgent issues that need to be explored before STAR-RISs can be harmoniously integrated into the existing SWIPT systems.
	
Firstly, thanks to the \emph{full}-space coverage, the distribution of served users can be more random and dispersed in a STAR-RIS assisted SWIPT system. However, this potential variation may magnify the channel differences among users, thus posing a more intractable resource allocation challenge to balance the performance between IUs and EUs. More importantly, the addition of STAR-RISs will introduce more optimization variables that are not considered in the conventional reflecting-only RIS assisted SWIPT system. Accordingly, the formulated optimization problems will be more intractable, which is beyond the capabilities of the methods applied in \cite{Multi_objective_IRS,Waveform,XU_large}. As such, efficient algorithms need to be developed to resolve the resource allocation problem for STAR-RIS assisted SWIPT systems.
	
Secondly, most of the aforementioned researches \cite{Wu_resource,Zhao_UAV,MEC} assumed that the perfect CSI of all channels was available to the base station (BS) or access point (AP). However, since there are more channels associated with passive STAR-RISs, more overhead and advanced channel estimation algorithms than conventional reflecting-only RISs need to be invested in the acquisition of accurate CSI. This is a very challenging but still infancy task \cite{Chanel_estimation}. Therefore, in order to reduce the impact of channel uncertainty on system performance in practical scenarios, it is necessary to investigate a robust design to deal with channel estimation errors.
	
Thirdly, the authors of \cite{Mu_star} proposed three different operating protocols for STAR-RIS and revealed their respective benefits for different communication scenarios and requirements. Interestingly, EUs and IUs have different power sensitivities to signal beams, transmission environments, and inter-user interference in the SWIPT system. Bearing all this in mind, a new but crucial question, i.e., \emph{Which protocol is more suitable for the SWIPT system?} is naturally raised. Note that the authors of \cite{STAR_SWIPT} only considered the ES protocol for STAR-RIS assisted SWIPT and did not answer this question, which is still an open issue and needs to be answered.
	
Motivated by these observations, we investigate the robust resource allocation for a STAR-RIS assisted SWIPT system based on the imperfect CSI assumption. In particular, two operating protocols, i.e., ES and TS, are considered. The main contributions of this paper are summarized as follows:
\begin{itemize}
	\item We propose a STAR-RIS assisted SWIPT system, where an STAR-RIS is deployed to assist a multi-antenna AP to simultaneously transmit information and power to two types of single-antenna users, i.e., IUs and EUs. From a practical perspective, we consider the imperfect CSI for all channels. In order to investigate the fundamental trade-off between IUs and EUs, we formulate MOOPs for both ES and TS to simultaneously maximize the minimum data rate for IUs and the minimum harvested power among EUs. 
	
	\item For ES, we first leverage the $\epsilon$-constraint method to transform the resulting MOOP into a more tractable SOOP. Then, we adopt the general S-procedure to approximate the semi-infinite inequality constraints for any given $\epsilon$. Finally, to solve the reformulated highly-coupled non-convex SOOP, we develop an efficient alternating optimization (AO) framework by solving the active AP beamforming and passive STAR-RIS beamforming in an iterative manner. In particular, a penalty-based method is leveraged to relax the rank-one constraint for STAR-RIS beamforming optimization. Besides, to ensure the feasibility of $\epsilon$-constraint method, we apply bisection search to obtain the performance upper boundary of IUs and determine the reasonable range for $\epsilon$.
		
	\item For TS, we still apply the $\epsilon$-constraint method and general S-procedure to transform the MOOP. Subsequently, we extend the proposed algorithm for ES into a two-layer optimization algorithm to solve the reformulated SOOP. In the outer-layer iteration, we determine the optimal time allocation via one-dimensional search. While for the inner-layer, we update the remaining variables by utilizing the AO-based algorithm in each iteration.
		
	\item Our numerical results depict that 1) deploying STAR-RISs can achieve a significant rate-energy region gain to conventional RISs in SWIPT systems, especially in terms of harvested power; 2) ES is more attractive for improving user fairness when focusing only on IUs or EUs, while TS is superior at balancing the performance between IUs and EUs. More importantly, TS can provide a more robust design under the imperfect CSI; 3) the deployment strategy of having all EUs on the same side of the STAR-RIS while all IUs on the other side yields higher performance gains, but also results in more implementation difficulties.
\end{itemize}
\subsection{Organization and Notations}
The rest paper is organized as follows: Section II introduces the system model and the joint beamforming MOOP formulations for both ES and TS protocols. In Section III, efficient solutions are proposed for robust resource allocation based on each operating protocol. Next, the numerical results are presented to verify the benefits of deploying the STAR-RIS in a SWIPT system in Section IV. Finally, Section V concludes the paper.
	
\emph{Notations}: Scalars, vectors, and matrices are denoted by lower-case, bold lower-case letters, and bold upper-case letters, respectively. $\|\mathbf{a}\| $ denotes the Euclidean norm of vector $\mathbf{a}$, $\mathbf{a}^H$ and $\mathbf{a}^T$ denotes the conjugate transpose and transpose of vector $\mathbf{a}$, respectively. For a square matrix $\mathbf{A}$, $\rm{Tr}(\mathbf{A})$, $\mathrm{Rank}(\mathbf{A})$, $\mathbf{A}^H$, $\mathbf{A}^T$ denote its trace, rank, conjugate transpose, and transpose, respectively, while $\mathbf{A} \succeq 0$ represents that $\mathbf{A}$ is a positive semidefinite matrix. And $\mathrm{Diag}\left(\mathbf{A}\right)$ denotes a vector whose elements are extracted from the main diagonal elements of matrix $\mathbf{A}$. $\|\mathbf{A}\|_*$, $\|\mathbf{A}\|_2$,  $\|\mathbf{A}\|_F$ denote the nuclear norm, spectral norm, and Frobenius norm of matrix $\mathbf{A}$, respectively. $\mathbb{C}^{M \times N}$ and $\mathbb{R}^{M\times N}$ denote the space of $M \times N$ complex valued matrices and real valued matrices respectively. $\mathbf{I}^{M\times M}$ denotes an identity matrix of size $M \times M$. $\otimes$ represents the Kronecker product. Finally, the distribution of a circularly symmetric complex Gaussian (CSCG) random variable with mean $\mu$ and variance $\sigma$ is denoted by $\mathcal{CN}(\mu, \sigma^2)$.
\section{System Model and Problem Formulation}
\begin{figure}
	\setlength{\abovecaptionskip}{0cm}   
	\setlength{\belowcaptionskip}{0.2cm}   
	\setlength{\textfloatsep}{7pt}
	\centering
	\includegraphics[width=3.5in]{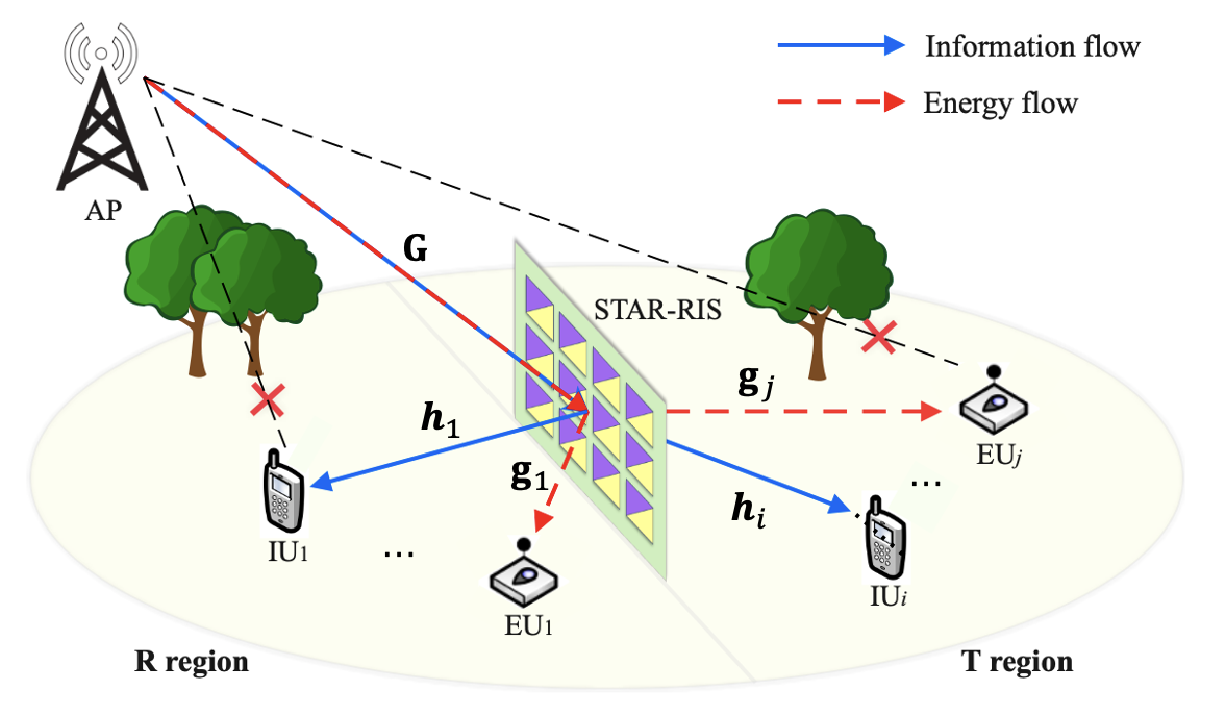}
	\caption{Illustration of the STAR-RIS assisted SWIPT systems}
	\label{system}
\end{figure}
As shown in Fig. \ref{system}, we consider a STAR-RIS assisted wireless system, where an STAR-RIS consisting of $M$ elements is deployed to facilitate SWIPT from the AP with $N$ antennas to two different sets of single-antenna users, i.e., IUs and EUs. The sets of IUs and EUs are denoted by $\mathcal{K_I}=\{1,\cdots,K_I\}$ and $\mathcal{K_E}=\{1,\cdots,K_E\}$, respectively. We assume these users are randomly located in both the reflection (R) and transmission (T) regions of the STAR-RIS. Without loss of generality, we assume the direct links between the AP and users are blocked by obstacles, and the communication for users in blind spots is supported by the STAR-RIS enabled transmission and reflection links \cite{Wu_resource,Zhao_UAV,MEC,STAR_SWIPT}. The quasi-static flat-fading model is assumed for all channels\footnote{In this paper, we consider the sub-6GHz frequencies and moderate size of AP and STAR-RIS  arrays \cite{Wu_resource,Zhang_STAR_NOMA}. Therefore, the simple far-field channel model is a safe approximation.}, and the channel coefficients from the AP to the STAR-RIS, from the STAR-RIS to IU $i$, and from the STAR-RIS to EU $j$ are denoted by $\mathbf{G}\in \mathbb{C}^{M\times N}$, $\mathbf{h}_{i}\in \mathbb{C}^{M\times 1}$, and $\mathbf{g}_{j}\in\mathbb{C}^{M\times 1}$, respectively. 
	
\subsection{STAR-RIS Protocols and Models}
In this paper, we consider the ES and TS protocols for STAR-RISs proposed in \cite{Mu_star}.

\emph{1) The ES protocol}: All elements of the STAR-RIS can split an incident signal into two independent signals, i.e., reflected and transmitted signals. Based on energy splitting, the beamforming matrix is divided into the reflection-coefficient part and the transmission-coefficient part, denoted by $\mathbf{\Theta}^{\textup{ES}}_r=\textup{diag}\left(\sqrt{\beta_1^{r}}e^{j\theta^r_1},\cdots,\sqrt{\beta_M^{r}}e^{j\theta^r_M}\right)$ and $\mathbf{\Theta}^{\textup{ES}}_t=\textup{diag}\left(\sqrt{\beta_1^{t}}e^{j\theta^t_1},\cdots,\sqrt{\beta_M^{t}}e^{j\theta^t_M}\right)$, respectively. Here, $\beta_m^{r},\beta_m^{t}$ $\in[0,1]$ and $\theta^r_m, \theta^t_m$ $\in[0,2\pi),\forall m\in \mathcal{M}=\{1,2,\cdots,M\}$ denote the amplitude and phase shift coefficients of the $m$-th element, respectively. Subject to the law of energy conservation, $\beta^{r}_m+\beta^{t}_m=1, \forall m\in\mathcal{M}$ is satisfied\footnote{In this paper, to determine the maximum performance gain, we investigate the beamforming design under the assumption of the phase shift coefficients for transmission and reflection can be adjusted independently. However, the proposed schemes can also be applied to the case where the phase shifts are coupled at the STAR-RIS with proper modifications \cite{Wang_coupled}.}. 

\emph{2) The TS protocol}: All elements of the STAR-RIS simultaneously complete the reflection/transmission of the signal in different orthogonal time durations. Therefore, the reflection-coefficient matrix and transmission-coefficient matrix are independent and given by
$\mathbf{\Theta}_r^{\textup{TS}}\!\!=\!\textup{diag}\left(e^{j\theta^r_1},\cdots,e^{j\theta^r_M}\right)$ and $\mathbf{\Theta}_t^{\textup{TS}}\!\!=\!\textup{diag}\left(e^{j\theta^t_1},\cdots,e^{j\theta^t_M}\right)$, respectively, where $\theta^r_m,\theta^t_m\in[0,2\pi), \forall m\in \mathcal{M}$. Let $\lambda_r$, $\lambda_t\in[0,1]$ denote the time allocation for reflection and transmission, the constraint $\lambda_r+\lambda_t=1$ holds for the whole communication period.

\subsection{Signal Transmission Model}
In this paper, we consider a linear transmit precoding at the AP and assume that each IU $i$ and EU $j$ are assigned with one dedicated information/energy beam, which are denoted by $\mathbf{w}_i$ and $\mathbf{v}_j$, respectively. Consequently, the transmitted signal at the AP is written as
\begin{align}
	\mathbf{x}=\sum_{i\in\mathcal{K_I}}\mathbf{w}_ix_i^I+\sum_{j\in\mathcal{K_E}}\mathbf{v}_jx_j^E,
\end{align}
where $x_i^I$ is the information-bearing signal for IU $i$, which is satisfied with $x_i^I \sim \mathcal{CN}(0,1), \forall i \in \mathcal{K_I}$, and $x_j^E$ is the energy-carrying signal for EU $j$ with $\mathbb{E}(|x_j^E|^2)=1, \forall j \in \mathcal{K_E}$. For simplification, we assume the information beam and energy beam are independent. Then we have
\begin{align}	\mathbb{E}\{\mathbf{x}^H\mathbf{x}\}=\sum_{i\in\mathcal{K_I}}\|\mathbf{w}_i\|^2+\sum_{j\in\mathcal{K_E}}\| \mathbf{v}_j \|^2 \leq P_{\max},
\end{align}
where $P_{\max}$ is the power budget at the AP.
	
\emph{1) ES}: When the ES protocol is employed, the received signal at IU $i$ is given by 
\begin{align}
	y_{i}^I=\mathbf{h}_{i}^{H}\mathbf{\Theta}^{\textup{ES}}_{s_i}\mathbf{G}\mathbf{x}+n_i, \forall i \in \mathcal{K_I},
\end{align}
where $s_i\in\{t,r\}$ indicates the region where IU $i$ is located, and $s_i=t$ if the user is located in T region, otherwise, $s_i=r$.
$n_i\sim \mathcal{CN}(0,\sigma^2)$ denotes the additive white Guassion noise at IU $i$. Since the energy beam is a Gaussian pseudo-random sequence and carries no information, we assume it can be decoded by the receiving IU and removed from the received signal, as properly adopted in \cite{Xu_model}. Therefore, the SINR of IU $i$ is given as
\begin{align}
	\textup{SINR}_i=\frac{|\left(\mathbf{u}_{s_i}^{\textup{ES}}\right)^{H}\mathbf{H}_i\mathbf{w}_i|^2}{\sum_{k\in\mathcal{K_I},k\ne i}|\left(\mathbf{u}_{s_i}^{\textup{ES}}\right)^{H}\mathbf{H}_i\mathbf{w}_k|^2+\sigma^2},
\end{align}
where $\mathbf{u}^{\textup{ES}}_{s_i}=[\sqrt{\beta_1^{s_i}}e^{j\theta^{s_i}_1},\cdots,\sqrt{\beta_M^{s_i}}e^{j\theta^{s_i}_M}]^H,{s_i}\in\{t,r\}$, $\mathbf{H}_i=\mathrm{diag}(\mathbf{h}_{i})\mathbf{G}$ denotes as the cascaded channel from the AP to IU $i$. 
	
On the other hand, due to the characteristics of energy broadcasting, all information and energy beams are desirable for each EU. At this point, the energy harvesting (EH) of EUs is virtually unaffected by interference and noise \cite{Xu_model}. Thus, based on the linear EH model\footnote{Although the non-linear EH model can capture the saturation effect more precisely for a single EH circuit \cite{NL_basic}, multi-parallel EH circuits effectively rectify the non-linear effect and cause a large linear conversion region \cite{Ma_NL}. Inspired by this, we adopt the linear EH model to investigate the fundamental design insights for STAR-RIS assisted SWIPT systems. Besides, by replacing the EH model with the non-linear model proposed in \cite{Xu_NL} and \cite{Zargari_NL}, our optimization framework is applicable, but some modifications still need to be investigated, which are our future focus topics.} \cite{Wu_Weighted,Chen_linear}, the received RF energy/power\footnote{In this paper, we use the unit time of 1 second to measure system performance. Thus, the terms “power” and “energy” are interchangeable.} at EU $j$ is denoted as $P_j$, given by 
\begin{align}
	P_j\!=\!\eta\!\left(\!\sum_{i\in\mathcal{K_I}}\!\!\big|\!\left(\mathbf{u}_{s_j}^{\textup{ES}}\right)^{H}\!\!\!\mathbf{G}_j\mathbf{w}_i\big|^2\!\!+\!\!\!\sum_{k\in\mathcal{K_E}}\!\!\big|\!\left(\mathbf{u}_{s_j}^{\textup{ES}}\right)^{H}\!\!\!\mathbf{G}_j\mathbf{v}_k\big|^2\!\!\right),\! \forall j \!\in\! \mathcal{K_E},
\end{align}
where $\mathbf{G}_j=\textup{diag}(\mathbf{g}_{j})\mathbf{G}$ denotes as the cascaded channel from the AP to EU $j$, $s_j \in\{t,r\}$ represents the located region for EU $j$. $\eta \in(0,1]$ represents the energy conversion efficiency for all EUs. For ease of analysis, we set $\eta=1$ in the subsequent context.

\emph{2) TS:} When the TS protocol is employed, the AP can communicate adaptively with users located in T and R regions in different time slots, so the interference among IUs from different regions can be ignored.
For ease of identification and description, we re-denote $\mathbf{w}^s_i$ and $\mathbf{v}^s_j, s \in \{t,r\}$ as the beamforming for IU $i$ and EU $j$, respectively. Note that $\mathbf{w}^r_i=0$ if IU $i$ is located in the T region, and $\mathbf{w}^t_i=0$ if IU $i$ is located in the R region. The same guidelines also apply to $\mathbf{v}^s_j$ for EU $j$. In this way, the SINR at IU $i$ and the harvested power of EU $j$ are rewritten as \eqref{TS1} and \eqref{TS2}, respectively, which is shown at the top of the next page,
\begin{figure*}[!t]
	\normalsize
	\begin{subequations}
		\begin{align}
			\label{TS1} \textup{SINR}^s_i \!=\!\frac{|\left(\mathbf{u}_{s}^{\textup{TS}}\right)^{H}\mathbf{H}_i\mathbf{w}^s_i|^2}{\sum_{k\in\mathcal{K_I},k\ne i}|\left(\mathbf{u}_{s}^{\textup{TS}}\right)^{H}\mathbf{H}_i\mathbf{w}^s_k|^2\!+\!\sigma^2}, \forall i\in \mathcal{K_I}, 
			s =\!\! \! \begin{cases}
				r,\text{if IU $i$ located in R region}, \\
				t,\text{if IU $i$ located in T region}.
			\end{cases} \\
			\label{TS2} P^s_j  \! \!=\!\!\!\sum_{i\in\mathcal{K_I}}\!\left|\left(\mathbf{u}_s^{\textup{TS}}\right)^{H}\!	\!\!\mathbf{G}_j\mathbf{w}^s_i\right|^2\!\!\!\!+\!\!\!\sum_{k\in\mathcal{K_E}}\!\!\left|\left(\mathbf{u}_s^{\textup{TS}}\right)^{H}\!\!\!\mathbf{G}_j\mathbf{v}^s_k\right|^2\! \!,\! \forall j  \! \in\! \mathcal{K_E}, \!
			s \!\!= \!\! \begin{cases}
				r,\text{if EU $j$ located in R region}, \\
				t,\text{if EU $j$ located in T region}.
			\end{cases}
		\end{align}
	\end{subequations}
	\hrulefill \vspace*{0pt}
\end{figure*}
where $\mathbf{u}_{s}^{\textup{TS}}=[e^{j\theta^{s}_1},\cdots,e^{j\theta^{s}_M}]^H,{s}\in\{t,r\}$.
	
\subsection{CSI Error Model}
Due to the passive property and channel complexity of the STAR-RIS, it is challenging to acquire the accurate CSI of the links associated with the STAR-RIS. To consider this effect on resource allocation design for the proposed system, we adopt a practical bounded CSI error model, where all possible channel errors are modeled by a bounded set \cite{Yu_Robust,Pan_framework}\footnote{In this paper, we exploit the flexible and general bounded CSI error model to characterize the CSI uncertainty caused by factors of the CSI acquisition process in the STAR-RIS assisted system. Under this model, the robust resource allocation for the worst case is designed, which with appropriate modifications can also be further extended for STAR-RIS assisted SWIPT adopting other CSI error models \cite{Pan_framework}.}. To be more specific, we directly rewrite the cascaded channels from the AP to IUs/EUs via the STAR-RIS with the following models:
\begin{gather}
	\mathbf{H}_i=\widehat{\mathbf{H}}_i+\triangle \mathbf{H}_i, \forall i \in \mathcal{K_I},\\
	\mathcal{H}_i \triangleq\{\triangle \mathbf{H}_i\in \mathbb{C}^{M\times N} :\|\triangle \mathbf{H}_i\|_{F} \leq \varepsilon_i\},\\
	\mathbf{G}_j=\widehat{\mathbf{G}}_j+\triangle \mathbf{G}_j, \forall j \in \mathcal{K_E},\\
	\mathcal{G}_j\triangleq\{\triangle \mathbf{G}_j\in \mathbb{C}^{M\times N} :\|\triangle \mathbf{G}_j\|_{F} \leq \mu_j\},
\end{gather}
where $\widehat{\mathbf{H}}_i,\forall  i \in \mathcal{K_I}$ and $\widehat{\mathbf{G}}_j,\forall j \in \mathcal{K_E}$ are estimations of the corresponding channel $\mathbf{H}_i$ and $\mathbf{G}_j$, respectively\footnote{Specially, the channel estimation for our proposed system can utilize the method applied in \cite{Wu_estimation,Xu_estimation}.}, $\triangle \mathbf{H}_i$ and $\triangle \mathbf{G}_j$ denote the channel estimation errors. Sets $\mathcal{H}_i$ and $\mathcal{G}_j$ collect all possible channel estimation errors, $\varepsilon_i$ and $\mu_j$ denote the maximum threshold for the norms of the CSI estimation error vectors $\triangle \mathbf{H}_i$ and $\triangle \mathbf{G}_j$, respectively.
	
\subsection{Problem Formulation}
In this paper, our goal is to simultaneously maximize the minimum data rate for IUs and harvested power for EUs to guarantee fairness among all users under imperfect CSI. For the conflict between two objectives, we apply a MOOP framework to elaborately design the robust resource allocation. More specifically, we consider the following MOOP formulations based on two different protocols for STAR-RIS operation, i.e., ES and TS.

\emph{1) MOOP Formulation for ES}: To begin with, we only focus on maximizing the minimum harvested power among EUs, by jointly optimizing the active beamforming $\{\mathbf{w}_i\}$, $\{\mathbf{v}_j\}$ and passive configuration $\{\mathbf{u}^{\textup{ES}}_s\}$. Under imperfect CSI, the robust design problem is formulated as
\begin{subequations}\label{P1}
	\begin{align}
		&\max_{{\{\mathbf{w}_i\},\{\mathbf{v}_j\},\{\mathbf{u}_{s}^{\textup{ES}}\}}} \min_j \ P_j \nonumber \\ 
		&\label{P1_C1}\quad \ \ {\rm s.t.} \sum_{i\in\mathcal{K_I}} \|\mathbf{w}_i\|^2+\sum_{j\in\mathcal{K_E}}\|\mathbf{v}_j\|^2 \leq P_{\max},\\ 
		&\label{P1_C2}\quad \ \quad \quad \theta^s_m \in [0,2\pi), \forall m \in \mathcal{M}, s\in\{t,r\},\\
		&\label{P1_C3}\quad \ \quad \quad \beta_m^t, \beta_m^r\in [0,1], \ \beta_{m}^t+\beta_{m}^r=1, \forall m\in \mathcal{M}.
	\end{align}
\end{subequations}
where constraint \eqref{P1_C1} represents the total transmit power budget at the AP. \eqref{P1_C2} and \eqref{P1_C3} are the constraints on phase shift and amplitude for STAR-RIS elements, respectively. Next, with imposing the same constraints, the minimum data rate maximization problem for IUs is formulated as
\begin{subequations}\label{P2}
	\begin{align}
		\nonumber &\  \max_{{\{\mathbf{w}_i\},\{\mathbf{v}_j\},\{\mathbf{u}_{s}^{\textup{ES}}\}}} \min_i  \log_2(1+\textup{SINR}_i)\\ 
		&\quad\quad \ \ {\rm s.t.}\ \eqref{P1_C1}-\eqref{P1_C3}.
	\end{align}
\end{subequations}
Subsequently, the MOOP for ES based on problems \eqref{P1} and \eqref{P2} is formulated in the following:
\begin{subequations}\label{P3}
	\begin{align}
		\textup{(Q1)}: &\ \max_{{\{\mathbf{w}_i\},\{\mathbf{v}_j\},\{\mathbf{u}_{s}^{\textup{ES}}\}}} \min_j \ P_j \nonumber \\ 
		\textup{(Q2)}: & \ \max_{{\{\mathbf{w}_i\},\{\mathbf{v}_j\},\{\mathbf{u}_{s}^{\textup{ES}}\}}} \min_i \ \log_2(1+\textup{SINR}_i)\\ 
		&\quad\quad \ \ {\rm s.t.}\ \eqref{P1_C1}-\eqref{P1_C3}.
	\end{align}
\end{subequations}
	
\emph{2) MOOP Formulation for TS}:
Unlike in ES, where all users can access the entire communication time, TS introduces time allocation variables $\lambda_r$ and $\lambda_t$ for different regional users. The STAR-RIS will only provide services to users located in the corresponding region during the allocated time. Therefore, the MOOP framework for TS is formulated as 
\begin{subequations}\label{P4}
	\begin{align}
		\textup{(Q3)}:\nonumber&\  \max_{{\{\mathbf{w}^s_i\},\{\mathbf{v}^s_j\},\{\mathbf{u}_s^{\textup{TS}}\}}} \min_{i} \ {\lambda_s}P^s_j\\
		\textup{(Q4)}:\nonumber&\  \max_{{\{\mathbf{w}^s_i\},\{\mathbf{v}^s_j\},\{\mathbf{u}_s^{\textup{TS}}\}}} \min_{j} \ {\lambda_s}\log_2(1+\textup{SINR}^s_i)\\ 
		&\label{P4_C1}\ \ {\rm s.t.}\ \lambda_r\left(\sum_{i\in\mathcal{K_I}} \|\mathbf{w}^r_i\|^2+\!\sum_{j\in\mathcal{K_E}}\|\mathbf{v}^r_j\|^2\right) \nonumber \\
		&\quad \quad+\!\lambda_t\left(\sum_{i\in\mathcal{K_I}} \|\mathbf{w}^t_i\|^2+\!\!\!\sum_{j\in\mathcal{K_E}}\|\mathbf{v}^t_j\|^2\right)\! \leq\! P_{\max},\\ 
		&\label{P4_C2}\quad \quad \ \theta^s_m \!\in\! [0,2\pi), \forall m \!\in\! \mathcal{M}, s\!\in\!\{t,r\},\\
		&\label{P4_C3}\quad \quad \ \lambda_r,\lambda_t \in [0,1], \ \lambda_r+\lambda_t=1.
	\end{align}
\end{subequations}
Similar to problem \eqref{P3}, optimization objectives $\textup{(Q3)}$ and $\textup{(Q4)}$ denote the minimum harvested power and minimum data rate maximization, respectively. \eqref{P4_C1} and \eqref{P4_C2} are the constraints of the power budget at AP and the phase shift of the STAR-RIS, respectively. Besides, \eqref{P4_C3} is a new constraint on time allocation introduced by the TS protocol.
	
Observe that both \eqref{P3} and \eqref{P4} are tricky problems that are challenging to solve directly. The main causes are as follows, briefly stated: 1) Since the two optimization objective functions of each problem are conflicting, there is no available resource allocation strategy to maximize the two goals simultaneously; 2) all optimization variables in the objective functions are closely coupled, which leads to the highly non-convexity for optimization problems; and 3) owing to the uncertainty of CSI, all objective functions have infinite possibilities, which are infeasible to handle in polynomial time. To sum up, the combination of the above factors makes it impossible to find an existing algorithm that can be directly applied to our formulated MOOPs. To circumvent this issue, we will explore efficient algorithms to solve the robust resource allocation problems \eqref{P3} and \eqref{P4} in the next. 
	
\section{Solutions of Robust Resource Allocation Problems}
In this section, we first propose an $\epsilon$-constraint based AO algorithm to design the joint beamforming for ES. Then, the algorithm is further extended to optimize the time allocation policy and beamforming vectors with a two-layer iterative method for TS.
\subsection{Proposed Solution for ES}
To start with, we investigate the robust resource allocation problem for ES by jointly optimizing the active beamforming (i.e., $\{\mathbf{w}_i\}$and $\{\mathbf{v}_j\}$) at the AP and passive beamforming (i.e., $\{\mathbf{u}^{\textup{ES}}_s\}$) of the STAR-RIS. Inspired by the fact that the $\epsilon$-constraint method can generate the whole Pareto frontier of the two conflict objective values by setting different $\epsilon$ \cite{epsilon_method}, we first adopt the $\epsilon$-constraint method \cite{IBFD}, which transforms $\textup{(Q2)}$ into a constraint associated with $\epsilon$ and specifies $\textup{(Q1)}$ as the unique objective function. Accordingly, the original MOOP is reduced to the form of SOOP as follows:
\begin{subequations}\label{P5}
	\begin{align}
		\nonumber&\  \max_{{ \{\mathbf{w}_i\},\{\mathbf{v}_j\},\{\mathbf{u}^{\textup{ES}}_s\}}} \min \{P_j | j \in \mathcal{K_E}\}\\ 
		&\label{P5_C1}\quad \quad\  {\rm s.t.} \ \log_2 \left(1\!+ \frac{|\left(\mathbf{u}^{\textup{ES}}_{s_i}\right)^H\mathbf{H}_i\mathbf{w}_i|^2}{\sum_{k\in\mathcal{K_I},k\ne i}|\left(\mathbf{u}^{\textup{ES}}_{s_i}\right)^H\mathbf{H}_i\mathbf{w}_k|^2\!+\!\sigma^2}\right)\!\!\geq \!\epsilon, \nonumber \\
		&\quad \quad\quad \quad	\ \forall i \in \mathcal{K_I}, {s_i}\in\{t,r\} ,\\
		&\label{P5_C2}\quad \quad\quad \quad \ \eqref{P1_C1}-\eqref{P1_C3},
	\end{align}
\end{subequations}
where constraint \eqref{P5_C1} indicates the minimum data rate requirements for IUs. It can be shown that the constraint will get equality at the optimal solution to problem \eqref{P5}. Otherwise, more communication resources can be shifted from IUs to EUs to improve the objective function while ensuring that \eqref{P5_C1} holds. Therefore, different trade-offs between IUs and EUs will result from different $\epsilon$ values. By solving a series of SOOPs corresponding to the ergodic value of $\epsilon$, the entire Pareto boundary of the trade-off region for original MOOP can be properly characterized.
	
However, the high coupling of variables and the infinite possibility of channel errors remain the main obstacles to solving SOOP \eqref{P5}. Let us define $\mathbf{W}_i \triangleq\mathbf{w}_i\mathbf{w}_i^{H}, \forall i \in \mathcal{K_I}$, $\mathbf{V}_j\triangleq\mathbf{v}_j\mathbf{v}_j^{H}, \forall j \in \mathcal{K_E}$ as well as $\mathbf{U}^{\textup{ES}}_s\triangleq\mathbf{u}^{\textup{ES}}_s\left(\mathbf{u}^{\textup{ES}}_s\right)^H$. Then, by introducing an auxiliary variable $\eta$, which satisfies $P_j \geq \eta, \forall j \in \mathcal{K_E}$, the considered problem \eqref{P5} can be transformed into the following equivalent rank-constrained SDP:
\begin{subequations}\label{P6}
	\begin{align}
		\nonumber&\  \max_{{\{\mathbf{W}_i\},\{\mathbf{V}_j\},\{\mathbf{U}^{\textup{ES}}_s\},\eta}} \;\eta\\ 
		&\label{P6_C1}\quad \ {\rm s.t.}\ \mathrm{Tr}\left(\left(\sum_{i\in\mathcal{K_I}}\mathbf{W}_i+\sum_{k\in\mathcal{K_E}}\mathbf{V}_k\right)\mathbf{G}^H_j\mathbf{U}^{\textup{ES}}_{s_j}\mathbf{G}_j\right)\geq \eta, \nonumber \\
		&\quad \quad\ \quad  \forall j \in \mathcal{K_E}, \forall {s_j}\in\{t,r\},\\
		&\label{P6_C2}\quad \quad\ \quad \mathrm{Tr}\left(\left(\frac{\mathbf{W}_i}{\Gamma} \!-\sum_{k\in\mathcal{K_I},k\ne i}\!\mathbf{W}_k\right)\mathbf{H}^H_i\mathbf{U}^{\textup{ES}}_{s_i}\mathbf{H}_i\right)\geq \sigma^2, \nonumber \\
		&\quad \quad\ \quad  \forall i \in \mathcal{K_I}, \forall {s_i}\in\{t,r\},\\
		&\label{P6_C3}\quad \quad\ \quad \sum_{i\in\mathcal{K_I}} \mathrm{Tr}(\mathbf{W}_i)+\sum_{j\in\mathcal{K_E}} \mathrm{Tr}(\mathbf{V}_j) \leq P_{\max},\\ 
		&\label{P6_C4}\quad \quad \ \quad \mathbf{W}_i \succeq 0, \mathrm{Rank}(\mathbf{W}_i)=1,  \forall i \in \mathcal{K_I},\\
		&\label{P6_C5}\quad \quad \ \quad \mathbf{V}_j \succeq 0, \mathrm{Rank}(\mathbf{V}_j)=1,  \forall j \in \mathcal{K_E}, \\
		&\label{P6_C6}\quad \quad\ \quad \mathbf{U}^{\textup{ES}}_s\succeq 0, \mathrm{Rank}\left(\mathbf{U}^{\textup{ES}}_s\right)=1, \forall s\in\{t,r\}, \\
		&\label{P6_C7}\quad \quad\ \quad \mathrm{Diag}   \left(\mathbf{U}^{\textup{ES}}_s\right)=\boldsymbol{\beta}^s, \forall s\in\{t,r\},\\
		&\label{P6_C8}\quad \quad\ \quad \beta_m^t, \beta_m^r\in [0,1], \beta_{m}^t+\beta_{m}^r=1, \forall m\in \mathcal{M},
	\end{align}
\end{subequations}
where $\Gamma=2^\epsilon-1$ denotes the minimum required SINR for IUs, and $\boldsymbol{\beta}^s \triangleq [\beta^s_1,\beta^s_2,\cdots,\beta^s_M], \forall s \in \{t,r\}$ denotes the amplitude adjustment vector. According to the identity $\mathrm{Tr}(\mathbf{A}^H\mathbf{B}\mathbf{C}\mathbf{D})=\textup{vec}(\mathbf{A})^{H}(\mathbf{D}^T\otimes\mathbf{B})\textup{vec}(\mathbf{C})$, \eqref{P6_C1} and \eqref{P6_C2} can be further expressed as
\begin{align}
	&\label{EU}\text{vec}(\mathbf{G}_j)^{H}\left(\left(\sum_{i\in\mathcal{K_I}}\mathbf{W}_i+\sum_{k\in\mathcal{K_E}}\mathbf{V}_k\right)^T \! \otimes\!\mathbf{U}^{\textup{ES}}_{s_j}\right) \textup{vec}(\mathbf{G}_j)\geq \eta, \nonumber \\
	&\quad \quad \quad \quad \forall j \in \mathcal{K_E}, \forall {s_j}\in\{t,r\}, \\
	&\label{IU}\text{vec}(\mathbf{H}_i)^{H}\left(\left(\frac{\mathbf{W}_i}{\Gamma}-\sum_{\substack{k \in \mathcal{K_I},\\k \neq i}}\!\mathbf{W}_k\right)^T \! \otimes\!\mathbf{U}^{\textup{ES}}_{s_i}\right) \textup{vec}(\mathbf{H}_i)\geq \sigma^2, \nonumber \\
	&\quad \quad \quad \quad \forall i \in \mathcal{K_I},  \forall {s_i}\in\{t,r\}.
\end{align}
Due to the uncertainty of $\mathbf{H}_i$ and $\mathbf{G}_j$, there are an infinite number of such constraints \eqref{EU} and \eqref{IU} in problem \eqref{P6}. Next, we will draw on the following lemma to deal with this issue.
\begin{lemma}\label{S-Procedure}
	\emph{(General S-Procedure \cite{convex}): Define the quadratic functions of the variable $\mathbf{x} \in \mathbb{C}^{N\times 1}$:
		\begin{align}
			f_i(\mathbf{x})=\mathbf{x}^H\mathbf{A}_i\mathbf{x}+2\textup{Re}\{\mathbf{x}^H\mathbf{b}_i\}+c_i, \forall i=1,\cdots,K, \nonumber
		\end{align}
		where $\mathbf{A}_i$ is the complex symmetric matrix, i.e., $\mathbf{A}_i=\mathbf{A}_i^{H}$. Then the $\{f_i(\mathbf{x})\geq 0\}_{i=1}^{K}\Rightarrow{f_0({\mathbf{x})\geq 0}}$ holds if and only if there exists $\forall i, \lambda_i \geq 0$ satisfying with 
		\begin{align}
			\begin{pmatrix}
				\mathbf{A}_0&\mathbf{b}_0\\
				\mathbf{b}_0^{H}&c_0
			\end{pmatrix}-
			\sum_{i=1}^{K}\lambda_i\begin{pmatrix}
				\mathbf{A}_i&\mathbf{b}_i\\
				\mathbf{b}_i^{H}&c_i
			\end{pmatrix} \succeq 0.\nonumber
		\end{align}
	}
\end{lemma}
Next, we denote $\mathbf{S}_1=\left(\sum_{i\in\mathcal{K_I}}\mathbf{W}_i+\sum_{k\in\mathcal{K_E}}\mathbf{V}_k\right)^{T}$. With $\mathbf{G}_j=\widehat{\mathbf{G}}_j+\triangle \mathbf{G}_j$, \eqref{EU} can be converted to the following equivalent form:
\begin{align}
	&\label{EU1}\text{vec}(\widehat{\mathbf{G}}_j+\triangle \mathbf{G}_j)^{H}\left(\mathbf{S}_1\otimes\mathbf{U}^{\textup{ES}}_{s_j}\right) \textup{vec}(\widehat{\mathbf{G}}_j+\triangle \mathbf{G}_j)-\eta \nonumber \\
	&=\textup{vec}(\triangle \mathbf{G}_j)^{H} \mathbf{A} \textup{vec}(\triangle \mathbf{G}_j)+2\mathbf{Re}\{\textup{vec}(\triangle \mathbf{G}_j)^{H}\mathbf{b}\}+c \geq 0,
\end{align}
where $\mathbf{A}=\mathbf{S}_1\otimes \mathbf{U}^{\textup{ES}}_{s_j}$,
$\mathbf{b}=\mathbf{A}\textup{vec}(\widehat{\mathbf{G}}_j)$, $c=\textup{vec}(\widehat{\mathbf{G}}_j)^{H}\mathbf{b}-\eta$. Besides, the following equation holds $\|\triangle \mathbf{G}_j\|_F \leq \mu_j \Rightarrow \|\textup{vec}(\triangle \mathbf{G}_j)\|_2 \leq \mu_j $. The uncertainty of CSI in $\textup{(10)}$ can be expressed as 
\begin{align}\label{EU2}
	-\textup{vec}(\triangle \mathbf{G}_j)^H \mathbf{I} \textup{vec}(\triangle \mathbf{G}_j)+\mu^2_j\geq 0.
\end{align}
	
It is noted that by considering $\textup{vec}(\triangle \mathbf{G}_j) $ as the variable $\mathbf{x}$ in \textbf{Lemma 1}, \eqref{EU1} and \eqref{EU2} can be combined into a linear matrix inequality (LMI) as follows:
\begin{align}\label{EU_transformation}
	\begin{pmatrix}
		\mathbf{S}_1 \otimes \mathbf{U}^{\textup{ES}}_{s_j} \!+\!\lambda_{1,j}\mathbf{I}&\left(\mathbf{S}_1 \otimes \mathbf{U}^{\textup{ES}}_{s_j}\right)\textup{vec}(\widehat{\mathbf{G}}_j)\\
		\textup{vec}(\widehat{\mathbf{G}}_j)^H\left(\mathbf{S}_1 \otimes \mathbf{U}^{\textup{ES}}_{s_j}\right) & \alpha_j
	\end{pmatrix} \succeq 0,
\end{align}
where $\alpha_j \! \!=\! \! \textup{vec}(\widehat{\mathbf{G}}_j)^{H}\!\left(\mathbf{S}_1 \otimes \mathbf{U}^{\textup{ES}}_{s_j}\right)\textup{vec}(\widehat{\mathbf{G}}_j)\!-\!\eta\!-\!\lambda_{1,j}\mu_j$, $\{\lambda_{1,j} \geq 0\}^{\mathcal{K_E}}_{j=1}$ are the auxiliary variables.
	
Similarly, by denoting $\mathbf{S}_{2,i}=\frac{\mathbf{W}_i}{\Gamma}-\sum_{\substack{k \in \mathcal{K_I},\\k \neq i}}\!\mathbf{W}_k$, \eqref{IU} is converted into a LMI form:
\begin{align}\label{IU_transformation}
	\begin{pmatrix}
		\mathbf{S}_{2,i} \otimes \mathbf{U}^{\textup{ES}}_{s_i} \!+\!\lambda_{2,i}\mathbf{I}&\left(\mathbf{S}_{2,i} \otimes \mathbf{U}^{\textup{ES}}_{s_i}\right)\textup{vec}(\widehat{\mathbf{H}}_i)\\
		\textup{vec}(\widehat{\mathbf{H}}_i)^H\left(\mathbf{S}_{2,i} \otimes \mathbf{U}^{\textup{ES}}_{s_i}\right) & \varrho_i
	\end{pmatrix} \!\succeq 0,
\end{align}
where $\varrho_i=\textup{vec}(\widehat{\mathbf{H}}_i)^{H}\left(\mathbf{S}_{2,i} \otimes \mathbf{U}^{\textup{ES}}_{s_i}\right)\textup{vec}(\widehat{\mathbf{H}}_i)\!-\!\sigma^2\!-\!\lambda_{2,i}\varepsilon_i$
and $\{\lambda_{2,i} \geq 0\}^{\mathcal{K_I}}_{i=1}$ are the auxiliary variables. Through these transformations above, problem \eqref{P6} is recasted as 
\begin{subequations}\label{P7}
	\begin{align}
		\nonumber&\  \max_{{\{\lambda_{1,j}\},\{\lambda_{2,i}\},\{\mathbf{W}_i\},\{\mathbf{V}_j\},\{\mathbf{U}^{\textup{ES}}_s\},\eta}}\;\eta\\ 
		&\quad \ {\rm s.t.}\ \textup{(21)}, \textup{(22)},  \eqref{P6_C3}-\eqref{P6_C8}.
	\end{align}
\end{subequations}
However, the optimization variables are still coupled in \eqref{EU_transformation} and \eqref{IU_transformation}, and it is challenging to optimize them simultaneously. To solve this issue, the AO method is applied to decompose problem \eqref{P7} into two subproblems, i.e., active beamforming design and passive beamforming design, which can be alternatively optimized.
	
\emph{1) Active Beamforming Design}: Here, we exclusively concentrate on active beamforming design. For given $\{\mathbf{U}^{\textup{ES}}_s\}$, original problem \eqref{P7} is simplified to
\begin{subequations}\label{P8}
	\begin{align}
		\nonumber&\  \max_{\{\lambda_{1,j}\},\{\lambda_{2,i}\}, \{\mathbf{W}_i\},\{\mathbf{V}_j\},\eta} \;\eta\\ 
		&\quad \ {\rm s.t.}\ 
		\textup{(21)}, \textup{(22)}, \eqref{P6_C3}-\eqref{P6_C5}.
	\end{align}
\end{subequations}
As can be observed, the non-convexity of rank-one constraints \eqref{P6_C4} and \eqref{P6_C5} restricts the solution of the new problem \eqref{P8}. To handle it, we employ semi-definite relaxation (SDR) and ignore the constraints \eqref{P6_C4} and \eqref{P6_C5} directly according to the following theorem. 
\begin{theorem}\label{Rank-one}
	\emph{The optimal solutions for the relaxed version of problem \eqref{P8}, i.e., without considering rank-one constraints, always satisfy ${\rm {Rank}}\left( {{{\mathbf{W}}^*_i}} \right) = 1,\forall i \in \mathcal{K_I}$ and $\sum_{j\in\mathcal{K_E}}{\rm {Rank}}\left( {{{\mathbf{V}}^*_j}} \right) \leq 1$ for feasible ${P_{\max}}>0$ and $\epsilon>0$}.
\end{theorem}
\begin{proof}
	Please refer to the Appendix.
\end{proof}
	
As a result, the relaxed problem \eqref{P8} is a standard SDP, and can be efficiently solved via existing convex solvers such as CVX \cite{cvx}.
	
\emph{2) Passive Beamforming Design}:
In this subproblem, we optimize the passive beamforming with fixed $\{\mathbf{W}_i\}$ and $\{\mathbf{V}_j\}$. This subproblem is reduced to
\begin{subequations}\label{P9}
	\begin{align}
		\nonumber&\  \max_{\{\lambda_{1,i}\},\{\lambda_{2,i}\}, \{\mathbf{U}^{\textup{ES}}_s\},\eta}\; \eta\\ 
		&\quad \quad \quad \ {\rm s.t.} \ \ \textup{(21)}, \textup{(22)}, \eqref{P6_C6}-\eqref{P6_C8}.
	\end{align}
\end{subequations}
Note that the difficulty in the design of STAR-RIS coefficients stems from the rank-one constraint \eqref{P6_C6}. Different from the relaxation of problem \eqref{P8}, we first rewrite this constraint as an equivalent form \cite{Yu_RANK}:
\begin{align}
	\mathrm{Rank}(\mathbf{U}^{\textup{ES}}_s)=1 \Leftrightarrow \|\mathbf{U}^{\textup{ES}}_s\|_{*}-\|\mathbf{U}^{\textup{ES}}_s\|_{2}=0, \forall s\in\{t,r\},
\end{align}
where $\|\mathbf{U}^{\textup{ES}}_s\|_{*}=\sum_i \sigma_i(\mathbf{U}^{\textup{ES}}_s)$ and  $\|\mathbf{U}^{\textup{ES}}_s\|_{2}=\sigma_1(\mathbf{U}^{\textup{ES}}_s)$ denote the nuclear norm and spectral norm of $\mathbf{U}^{\textup{ES}}_s$, respectively, and $\sigma_i$ is the $i$-th largest singular value of $\mathbf{U}^{\textup{ES}}_s$. Then, the penalty method \cite{penalty} is leveraged, where we convert the constraint \eqref{P6_C6} into a non-negative penalty function term appended to the objective function as follows:
\begin{subequations}\label{P10}
	\begin{align}
		\nonumber&\  \max_{\{\lambda_{1,j}\},\{\lambda_{2,i}\}, \{\mathbf{U}^{\textup{ES}}_s\},\eta} \;\eta-\xi \sum_{s\in\{t,r\}} \left(\|\mathbf{U}^{\textup{ES}}_s\|_{*}-\|\mathbf{U}^{\textup{ES}}_s\|_{2}\right)\\\ 
		&\quad \quad \quad \ {\rm s.t.}\ \ \textup{(21)}, \textup{(22)}, \eqref{P6_C7}, \eqref{P6_C8},
	\end{align}
\end{subequations}
where $\xi>0$ is the penalty factor. Driven by the optimization goal, the value of the penalty term will decrease with the increase of $\xi$, and when $\xi \rightarrow +\infty$, the penalty term gradually converges to 0. At this point, the optimal solution $\mathbf{U}^{\textup{ES}}_s$ of problem \eqref{P10} always satisfies the equality constraint \eqref{P6_C6}. However, it should be noted that the initial value of $\xi$ has a significant impact on the effect of the algorithm. At the start of the iteration, if $\xi$ is too large, the focus of optimization will shift from the original function $\eta$ to the penalty item, which violates our intention. Hence, we need to initialize $\xi$ with a small value, and then gradually increase it until the convergence criterion of the rank-one constraint is met:
\begin{align}
	\max\{\|\mathbf{U}^{\textup{ES}}_s\|_{*}-\|\mathbf{U}^{\textup{ES}}_s\|_{2}, s\in\{t,r\}\} \leq \epsilon_1,
\end{align} 
where $\epsilon_1$ is a predefined maximum violation of the equality constraint.
	
Nevertheless, the non-convexity of the penalty term makes the reformulated problem \eqref{P10} still difficult to solve. Inspired by the successive convex approximation (SCA) technique, we approximate the penalty term by its first-order Taylor expansion to obtain the convex upper bound as $\textup{(29)}$, which is shown at the top of the next page, 
\begin{figure*}[!t]
	\normalsize
	\begin{equation}
		\begin{aligned}
			\|\mathbf{U}^{\textup{ES}}_s\|_{*}\!\!-\!\|\mathbf{U}^{\textup{ES}}_s\|_{2}
			&\leq \|\mathbf{U}^{\textup{ES}}_s\|_{*}\!\!-\!\left\{\!\|\mathbf{U}^{\textup{ES}(l)}_s\|_2\!+\!\mathrm{Tr}\!\left[\varphi \left(\mathbf{U}^{\textup{ES}(l)}_s\right)\varphi \left(\mathbf{U}^{\textup{ES}(l)}_s\right)^{H}\!\!\left(\mathbf{U}^{\textup{ES}}_s\!\!-\! \mathbf{U}_s^{\textup{ES}(l)}\right)\right]\!\right\} \\
			&\triangleq \|\mathbf{U}^{\textup{ES}}_s\|_{*}-\overline{\mathbf{U}}^{\textup{ES}(l)}_s,
		\end{aligned}
	\end{equation}
	\hrulefill \vspace*{0pt}
\end{figure*}
where $\varphi\left(\mathbf{U}^{\textup{ES}(l)}_s\right)$ denotes the eigenvector related to the largest eigenvalue of $\mathbf{U}^{\textup{ES}(l)}_s$.
Next, by replacing the penalty term with its convex upper bound according to the given point $\{{\mathbf{U}}^{\textup{ES}(l)}_s\}$, problem \eqref{P10} is approximated into the following optimization problem:
\begin{subequations}\label{P11}
	\begin{align}
		\nonumber&\  \max_{\{\lambda_{1,j}\},\{\lambda_{2,i}\}, \{\mathbf{U}^{\textup{ES}}_s\},\eta}\; \eta-\xi\sum_{s\in\{t,r\}}\left(\|\mathbf{U}^{\textup{ES}}_s\|_{*}-\overline{\mathbf{U}}^{\textup{ES}(l)}\right)\\ 
		&\quad \quad \quad\ {\rm s.t.}\ \ \textup{(21)}, \textup{(22)}, \eqref{P6_C7}, \eqref{P6_C8}.
	\end{align}
\end{subequations}
Now, problem \eqref{P11} is a SDP and can be solved by the CVX. Based on this, solving the problem \eqref{P11} repeatedly by exploiting SCA and updating the penalty factor with $\xi=\tau\xi$
until $\textup{(28)}$ is satisfied, we can obtain the optimal passive beamforming in each AO iteration.
And the penalty-based algorithm is summarized in \textbf{Algorithm 1}.
\begin{algorithm}[!t]
	\caption{penalty-based algorithm for solving problem \eqref{P11}}
	\begin{algorithmic}[1]
		\STATE {Initialize the feasible {$\mathbf{U}_s^{\textup{ES}(0)}$}, penalty factor $\xi$ and given $\{\mathbf{W}_i, \mathbf{V}_j\}$}, set the allowable $\epsilon_1$, $\epsilon_2$, and the maximum number of iterations $L_{\max}.$
		\STATE $\textbf{repeat}$
		\STATE \quad Set iteration index $l = 0$;
		\STATE \quad $\textbf{repeat}$
		\STATE \quad \ \ For given $\{\mathbf{U}^{\textup{ES}(l)}_s\}$, solve the problem \eqref{P11};
		\STATE \quad \ \ Update $\{\mathbf{U}^{\textup{ES}(l+1)}_s\}$ with the obtained optimal solutions, $l=l+1$;
		\STATE \quad \textbf{until} the iterative gain of the objective function value is below a predefined threshold $\epsilon_2 > 0$ or $l = L_{\max}$.
		\STATE \quad Update $\{\mathbf{U}^{\textup{ES}(0)}_s\}$ with the optimized solutions $\{\mathbf{U}^{\textup{ES}(l)}_s\}$.
		\STATE \quad Update $\xi=\tau\xi$.
		\STATE \textbf{until} the constraint violation is below a predefined threshold $\epsilon_1> 0$.
	\end{algorithmic}
\end{algorithm}
	
To this end, the original problem \eqref{P7} with highly coupled variables was decomposed into two subproblems, i.e., problems \eqref{P8} and \eqref{P11}, which are solved in an iterative manner, according to the AO method. Wherein, the non-convexity for active beamforming design is relaxed according to \textbf{Theorem 1}, while the rank-one constraint in STAR-RIS beamforming is solved by employing penalty-based \textbf{Algorithm 1}. On the one hand, the objective function of problem \eqref{P11} will gradually converge with the increase of the penalty factor in each passive beamforming design \cite{Yu_RANK}. On the other hand, the AO iteration algorithm is guaranteed to converge, and the relevant proofs can be found in the literature \cite{Wu_QoS}. Therefore, our proposed algorithm for problem \eqref{P7} will eventually converge to a stationary point. The specific details of the developed algorithm are presented in \textbf{Algorithm 2}.
	
In addition, there is a straightforward observation that the obtained solution to problem \eqref{P7} is generally sensitive to the value of $\epsilon$. In order to ensure the feasibility of problem \eqref{P7}, the value of $\epsilon$ should be in a reasonable range $[0,R_{\max}]$. $R_{\max}$ is the maximum achievable data rate for this system, and might be attained by resolving the subsequent optimization problem:
\begin{subequations}\label{P12}
	\begin{align}
		\nonumber&\max_{{\{\mathbf{w}_i\},\{\mathbf{v}_j\},\{\mathbf{u}_s^{\textup{ES}}\},\gamma}} \gamma\\ 
		&\label{P12_C1}\quad \quad \ {\rm s.t.} \quad \frac{|\left(\mathbf{u}^{\textup{ES}}_{s_i}\right)^H\mathbf{H}_i\mathbf{w}_i|^2}{\sum_{k\in\mathcal{K_I},k\ne i}|\left(\mathbf{u}^{\textup{ES}}_{s_i}\right)^H\mathbf{H}_i\mathbf{w}_k|^2+\sigma^2}\geq \gamma, \nonumber\\
		&\quad\quad \quad \quad \quad\forall i \in \mathcal{K_I}, s_i\in\{t,r\},\\ 
		&\label{P12_C2}\quad\quad \quad \quad \quad \eqref{P1_C1}-\eqref{P1_C3}.
	\end{align}
\end{subequations}
where $\gamma$ is an introduced auxiliary variable with $\gamma \geq 0$. Next, in order to simplify the fractional constraint \eqref{P12_C1}, we construct the following function:
\begin{align}
	f(\gamma)\!=\!{|\left(\mathbf{u}^{\textup{ES}}_{s_i}\right)^H \! \! \mathbf{H}_i\mathbf{w}_i|^2}\!\!-\!\!\gamma\left(\!\sum_{k\in\mathcal{K_I},k\ne i}\!\!|\left(\mathbf{u}^{\textup{ES}}_{s_i}\right)^H \! \! \mathbf{H}_i\mathbf{w}_k|^2\!\!+\!\sigma^2\!\right).
\end{align}
As can be seen, this function is monotonically decreasing with respect to $\gamma$ when holding other variables constant. Hence, we can turn the goal of maximizing $\gamma$ into finding the maximum value of the function $f(\gamma)$ and use the bisection method to find the optimal $\gamma$. When the optimized maximum value of $f(\gamma)>0$, we can further increase the value of $\gamma$ while ensuring that the original problem is feasible. On the contrary, if the optimized maximum value of $f(\gamma)<0$, we can only make the original problem feasible by reducing the value of $\gamma$. Thus, only the $\gamma$ that makes the optimized maximum value of $f(\gamma)=0$ is the optimal solution for the original problem. Considering that the non-smoothness of $f(\gamma)$ is difficult to handle, we introduce a new auxiliary variable $t_0$ and yield a new optimization problem:
\begin{algorithm}[!t]
	\caption{AO algorithm for solving problem \eqref{P7}}
	\begin{algorithmic}[1]
		\STATE{Determine the value of $\epsilon$ and initialize feasible point $\{\mathbf{U}^{\textup{ES}(0)}_s\}$} with random matrix and the iteration index $k=0$, set the allowable $\varepsilon_0$, the maximum number of iterations $K_{\max}$.
		\STATE \textbf{repeat}
		\STATE \quad{For given $\{\mathbf{U}^{\textup{ES}(k)}_s\}$, solve the problem \eqref{P8} to obtain the optimized $\{\mathbf{W}^{(k)}_i, \mathbf{V}^{(k)}_j\}$};
		\STATE \quad{Update $\{\mathbf{U}^{\textup{ES}(k+1)}_s\}$ via solving the problem \eqref{P11} by applying \textbf{Algorithm 1}};
		\STATE \quad Update $k=k+1$;
		\STATE  {\bf until} the iterative gain of the objective function value of problem \eqref{P7} is below a predefined threshold $\varepsilon_0 >0$ or $k=K_{\max}$.
	\end{algorithmic}
\end{algorithm}
\begin{subequations}\label{P13}
	\begin{align}
		\nonumber&\max_{{\{\mathbf{w}_i\},\{\mathbf{v}_j\},\{\mathbf{u}^{\textup{ES}}_s\},t_0}} t_0\\ 
		&\label{P13_C1}\ \ {\rm s.t.}  \	{|\!\left(\mathbf{u}^{\textup{ES}}_{s_i}\right)^H \! \!\mathbf{H}_i\mathbf{w}_i|^2}\!\!-\!\!\gamma\left(\!\sum_{k\in\mathcal{K_I},k\ne i} \! \! \!\! \!|\!\left(\mathbf{u}^{\textup{ES}}_{s_i}\right)^H \! \!\mathbf{H}_i\mathbf{w}_k|^2\!+\! \sigma^2\!\right)\! \! \geq \! t_0,\nonumber \\
		&\quad \quad \ \forall i\in \mathcal{K_I}, s_i \in\{t,r\},\\
		&\label{P13_C2}\quad \quad  \ \eqref{P1_C1}-\eqref{P1_C3}.
	\end{align}
\end{subequations}
	
In fact, given $\gamma$, problem \eqref{P13} can also be solved directly by applying \textbf{Algorithm 2}. Along this line, we can update the search range $[\gamma_{\min},\gamma_{\max}]$ in the $n$-th iteration and determine  $\gamma_{n+1}=\frac{\gamma_{\max}+\gamma_{\min}}{2}$ for the next search. According to the bisection criterion, the value of $\gamma$ will eventually converge when the following inequality is satisfied:
\begin{align}
	\|\gamma_n-\gamma_{n+1}\| \leq \varepsilon_2,
\end{align}
where $\varepsilon_2$ is a predefined threshold. After solving the above, we determine the maximum feasible value of $\epsilon$, i.e., $R_{\max}=\log_{2}\left(1+\gamma^*\right)$, for problem \eqref{P7}. Until here, the overall algorithm for the original MOOP \eqref{P3} can be completely outlined in \textbf{Algorithm 3}.
\begin{algorithm}[!t]
	\caption{$\epsilon$-constraint method solving MOOP \eqref{P3}}
	\begin{algorithmic}[1]
		\STATE {Solve problem \eqref{P12} by applying bisection search method and \textbf{Algorithm 2} to get the $R_{\max}$}, initialize the factor $\delta=0$ and step $\triangle\delta$.
		\STATE \textbf{repeat}
		\STATE \quad {Update the $\epsilon$ with $\epsilon=\delta R_{\max}$; }
		\STATE \quad {Given $\epsilon$, reformulate MOOP \eqref{P3} into a tractable SOOP \eqref{P7} via exploiting $\epsilon$-constraint method and S-procedure;}
		\STATE \quad {Solve problem \eqref{P7} by applying $\textbf{Algorithm 2}$ to obtain the optimized $\eta$ corresponding to given $\epsilon$};
		\STATE \quad {Update $\delta=\delta+\triangle\delta$ ;}
		\STATE {\bf until} $\delta>1$.
	\end{algorithmic}
\end{algorithm}
	
\subsection{Proposed Solution for TS}
In this subsection, \textbf{Algorithm 2} is extended to solve MOOP \eqref{P4} for TS. Similar to ES, we aim to maximize the minimum harvested power by EUs, subject to the minimum achievable data rate constraint in the reformulated SOOP. The difference is that new variables $\lambda_t$ and $\lambda_r$ for the time allocation need to be optimized. Let $\mathbf{W}^s_i \triangleq\mathbf{w}^s_i({\mathbf{w}^s_i})^H$, $\mathbf{V}^s_j\triangleq\mathbf{v}^s_j({\mathbf{v}^s_j})^H$ and $\mathbf{U}_s^{\textup{TS}}\triangleq\mathbf{u}_s^{\textup{TS}}\left({\mathbf{u}_s^{\textup{TS}}}\right)^H$. By applying the $\epsilon$-constraint method and \textbf{Lemma 1}, MOOP \eqref{P4} is transformed into a SOOP as 
\begin{subequations}\label{P14}
	\begin{align}
		\nonumber&\  \max_{{\{\mathbf{W}^s_i\},\{\mathbf{V}^s_j\},\{\mathbf{U}^{\textup{TS}}_s\},\{\lambda_s\},\eta}} \eta \\ 
		&\label{P14_C1}{\rm s.t.}\  \lambda_s\mathrm{Tr}\left(\left(\sum_{i\in\mathcal{K_I}}\mathbf{W}^s_i+\sum_{k\in\mathcal{K_E}}\mathbf{V}^s_k\right)\mathbf{G}^H_j\mathbf{U}^{\textup{TS}}_{s}\mathbf{G}_j\right)\geq \eta,\nonumber \\
		&\quad \ \ \forall j \in \mathcal{K_E}, \forall {s}\in\{t,r\},\\
		&\label{P14_C2}\quad \ \ \mathrm{Tr}\left(\left(\frac{\mathbf{W}^s_i}{2^{\frac{\epsilon}{\lambda_s}}-1} \!-\!\!\!\sum_{k\in\mathcal{K_I},k\ne i}\!\mathbf{W}^s_k\right)\mathbf{H}^H_i\mathbf{U}^{\textup{TS}}_{s}\mathbf{H}_i\right)\geq \sigma^2,\nonumber \\
		& \quad \ \  \forall i \in \mathcal{K_I}, \forall {s}\in\{t,r\},\\
		&\label{P14_C3}\quad \ \ \lambda_r(\sum_{i\in\mathcal{K_I}} \!\mathbf{W}^r_i \!+\!\!\sum_{j\in\mathcal{K_E}}\!\!\mathbf{V}^r_j)+\lambda_t(\sum_{i\in\mathcal{K_I}} \!\!\mathbf{W}^t_i \!+\!\!\sum_{j\in\mathcal{K_E}}\!\!\mathbf{V}^t_j) \leq P_A,\\ 
		&\label{P14_C4}\quad \ \ \mathbf{W}^s_i \succeq 0, \mathrm{Rank}(\mathbf{W}^s_i)=1,  \forall i \in \mathcal{K_I}, \forall s\in \{t,r\},\\
		&\label{P14_C5}\quad \ \ \mathbf{V}^s_j \succeq 0, \mathrm{Rank}(\mathbf{V}^s_j)=1,  \forall j \in \mathcal{K_E}, \forall s \in\{t,r\},\\
		&\label{P14_C6}\quad\ \ \mathbf{U}^{\textup{TS}}_s\succeq 0, \ \mathrm{Rank}\left(\mathbf{U}^{\textup{TS}}_s\right)=1, \forall s\in\{t,r\}, \\
		&\label{P14_C7}\quad\ \ \mathrm{Diag} \left(\mathbf{U}^{\textup{TS}}_s\right)=\mathbf{I}^{M\times M}, \forall s\in\{t,r\},\\
		&\label{P14_C8}\quad \ \  \lambda_r,\lambda_t \in [0,1],  \lambda_r+\lambda_t=1.
	\end{align}
\end{subequations}
It is intuitive to see that problem \eqref{P14} for given $\{\lambda_t,\lambda_r\}$ can be considered as a simpler form of problem \eqref{P6} without amplitude adjustment constraint, which can be solved by resorting to \textbf{Algorithm 2}. 
Therefore, the main difficulty in solving problem \eqref{P14} falls in the determination of the optimal time allocation policy. To overcome this issue, a two-layer algorithm is proposed. In the outer layer, the one-dimensional search is used to obtain the optimal time allocation $\{\lambda^*_t,\lambda^*_r\}$. While the remaining variables can be optimized by the inner-layer iteration with fixed $\{\lambda^*_t,\lambda^*_r\}$.
	
Likewise, in order to ensure the feasibility of the problem \eqref{P14}, we explore the maximum value of $\epsilon$ allowed for TS by invoking the similar method for ES. Specially, we first introduce a new variable $\gamma_0$ that satisfies $\lambda_s\log_2(1+\textup{SINR}^s_i) \geq \gamma_0, s\in\{t,r\}$ and build a new function as $\textup{(36)}$, which is shown at the top of the next page.
\begin{figure*}[!t]
	\normalsize
	\begin{align}
		f(\gamma_0)\!=\!{|\left(\mathbf{u}^{\textup{TS}}_s\right)^H \! \!\mathbf{H}_i\mathbf{w}^s_i|^2}\!-\!\left(2^{\frac{\gamma_0}{\lambda_s}}\!-\!1\right)\left(\sum_{k\in\mathcal{K_I},k\ne i} \! \! \!\! |\left(\mathbf{u}^{\textup{TS}}_{s}\right)^H \! \!\mathbf{H}_i\mathbf{w}^s_k|^2\!+\sigma^2\right), \forall i \in \mathcal{K_I}, s\in\{t,r\}.
	\end{align}
	\hrulefill \vspace*{0pt}
\end{figure*}
Then, a new optimization problem with introducing auxiliary variable $T_0$ is formulated as  
\begin{subequations}\label{P15}
	\begin{align}
		\nonumber&\max_{{\{\mathbf{w}^s_i\},\{\mathbf{v}^s_j\},\{\mathbf{u}^{\textup{TS}}_s\},T_0}} T_0\\ 
		&\ \ {\rm s.t.}  \quad f(\gamma_0) \geq \! T_0, \forall i\in \mathcal{K_I}, n \in\{t,r\}.\\
		&\quad \quad \quad \eqref{P4_C1}-\eqref{P4_C3}.
	\end{align}
\end{subequations}
It is obvious that the problem \eqref{P15} can be solved similarly to how the problem \eqref{P14} was previously presented. Accordingly, we can update the value of $\gamma_0$ with bisection criteria until $\textup{(34)}$ is satisfied. As a result, by traversing the value of $\epsilon$ in this determined feasible region and solving the related SOOP, we can get a proper rate-energy trade-off region for TS. Since it is similar to \textbf{Algorithm 3}, the specific details of the algorithm proposed for TS are omitted here.

\subsection{Computational Complexity and Convergence Analysis}
The computational complexity of the proposed algorithms is analyzed in the following. Since \textbf{Algorithm 3} for obtaining the entire rate-energy region can be seen as multiple iterations of \textbf{Algorithm 2} with different $\epsilon$, its computational complexity is $\left(I_{\epsilon}+I_b\right)$ times that of \textbf{Algorithm 2}, where $I_{\epsilon}$ and $I_{b}$ denote the ergodic number of the $\epsilon$-constraint method and the number of bisection search for obtaining $R_{\max}$, respectively. Thus, the computational complexity of \textbf{Algorithm 2} becomes the focus of attention. Note that the transformed subproblems are both standard SDP problems.
According to \cite{SDR}, the approximate computational complexity of \eqref{P8} and \eqref{P11} is given by $\mathcal{O}_A=\mathcal{O}\left((K_E+K_I)(M^{3.5}N^{3.5}+N^{3.5})\right)$ and $\mathcal{O}_P=\mathcal{O}\left((K_E+K_I)(M^{3.5}N^{3.5})+2M^{3.5})\right)$, respectively. Let $I^{\textup{ES}}_{A}$ and $I^{\textup{ES}}_{P}$ denote the iteration numbers of the AO and penalty method, respectively. The overall complexity for  \textbf{Algorithm 2} is $\mathcal{O}\left(I^{\textup{ES}}_{A}(\mathcal{O}_A+I^{\textup{ES}}_{P}\mathcal{O}_P)\right)$. 
	
Based on this, the complexity of the proposed algorithm for ES is measured as $\mathcal{O}^{\textup{ES}}=\left(I_\epsilon+I_b\right)\mathcal{O}\left(I^{\textup{ES}}_A\left(\mathcal{O}_A+I^{\textup{ES}}_P\mathcal{O}_P\right)\right)$. While for TS, one-dimensional search is used to find the optimal solution for time allocation. Let $L$ denote the search times, and then the complexity of the proposed algorithm for TS is measured as
$\mathcal{O}^{\textup{TS}}=L\left(I_\epsilon+I_b\right)\mathcal{O}\left(I^{\textup{TS}}_A\left(\mathcal{O}_A+I^{\textup{TS}}_P\mathcal{O}_P\right)\right)$, where $I^{\textup{TS}}_A$ and $I^{\textup{TS}}_P$ denote the iteration numbers of the AO and penalty method for TS, respectively. It can be seen that the complexity of the proposed TS algorithm is higher than that of ES due to the addition of one-dimensional search.

Note that the convergence of $\epsilon$-constraint method, AO, SCA, and penalty method has been proved in \cite{epsilon_method,cvx,penalty}. Based on this, with predetermined convergence criterion, the proposed algorithm for ES will eventually converge a stationary point via many iterations. Additionally, one-dimensional search is also a method that guarantees convergence. Thus, the proposed algorithm for TS is also capable of reaching convergence.
	
\section{Numerical Results}
Based on various operating protocols, numerical results obtained from different perspectives are presented to demonstrate the efficacy of the STAR-RIS on SWIPT systems.
\subsection{Simulation Setup}
\begin{figure}
	\setlength{\abovecaptionskip}{0cm}   
	\setlength{\belowcaptionskip}{0.2cm}   
	\setlength{\textfloatsep}{7pt}
	\centering
	\includegraphics[width=3.6in]{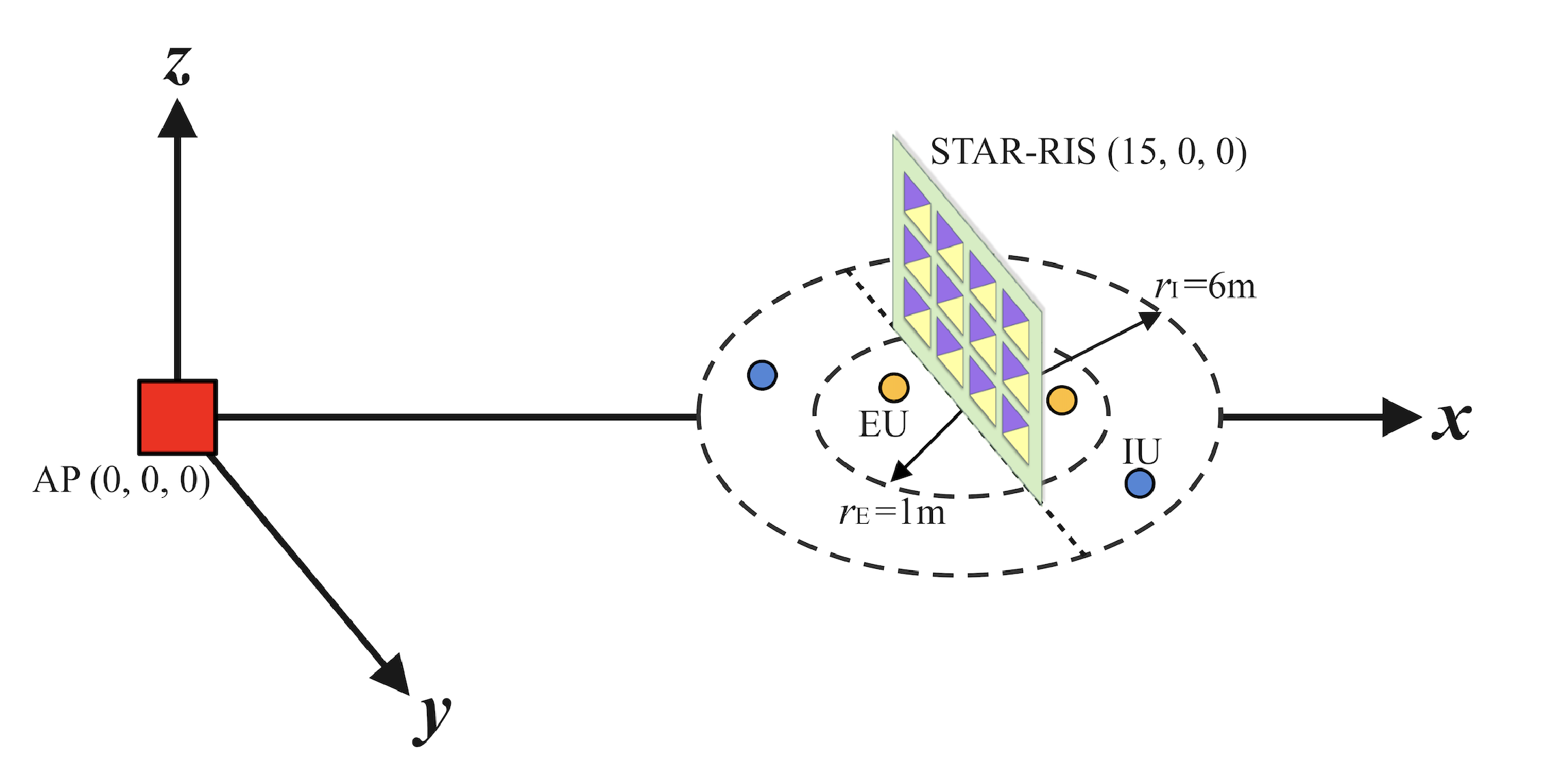}
	\caption{Simulation setup}
	\label{Setup}
\end{figure}
First we introduce the three-dimensional coordinate setup considered in our simulation. As shown in Fig. \ref{Setup}, the AP is located at $(0,0,0)$ meters, and the STAR-RIS with elements configured in a uniform planar array (UPA) is deployed at the user hotspot with coordinates $(15,0,0)$ meters. The EUs and IUs are randomly distributed in circle and ring-shaped areas centered at the STAR-RIS with a radius of $r_{E}=1$ and $r_{I}\in(1,6)$ m, respectively. For ease of illustration, we analyze a basic scenario where each side of the STAR-RIS is distributed with only one IU and one EU. In this paper, all channels are modeled as Rician fading channels as follows:
\begin{subequations}
	\begin{gather}
		\mathbf{G}=\sqrt{\rho_{AS}(d)}\left(\sqrt{\frac{K_{AS}}{K_{AR}+1}}\mathbf{G}^\textup{LoS}+\sqrt{\frac{K_{AS}}{K_{AS}+1}}\mathbf{G}^\textup{NLoS}\right),\\
		\mathbf{v}_k=\sqrt{\rho_{SU}(d)}\left(\sqrt{\frac{K_{SU}}{K_{SU}+1}}\mathbf{v}_k^\textup{LoS}+\sqrt{\frac{K_{SU}}{K_{SU}+1}}\mathbf{v}_k^\textup{NLoS}\right), \nonumber\\
		\mathbf{v} \in \{\mathbf{h},\mathbf{g}\}.
	\end{gather}
\end{subequations}
where $\rho_{AS}(d)=\frac{\rho_0}{d^{\alpha_{AS}}_{AS}}$, $\rho_{SU}(d)=\frac{\rho_0}{d^{\alpha_{SU}}_{SU}}$, and $\rho_0$ represents the path loss at reference 1 m, $d_{AS}$ and $d_{SU,k}$ denote the distance from the AP to the STAR-RIS and from  the STAR-RIS to the $k$-th user, respectively, $\alpha_{AS}$ and $\alpha_{SU}$ denote the corresponding path loss exponents. In addition, $K_{AS}$ and $K_{SU}$ are the Rician factors, and $\mathbf{G}^{\textup{LoS}}$ and $\mathbf{v}^{\textup{LoS}}_k$ are the corresponding deterministic LoS components, while $\mathbf{G}^{\textup{NLoS}}$ and $\mathbf{v}^{\textup{NLoS}}_k$ are the corresponding deterministic NLoS components, which are modeled as random Rayleigh fading components.  In addition, the normalized maximum channel estimation errors of IU $i$ and EU $j$ are set to be $\rho_{H}=\frac{\varepsilon_i}{\|\widehat{\mathbf{H}}_i\|}$ and 
$\rho_{G}=\frac{\mu_j}{\|\widehat{\mathbf{G}}_j\|}$, respectively. The specific system parameters are presented in Table I \cite{Mu_star,Multi_objective_IRS}.
\begin{table*}[t]\small
	\centering
	\caption{\textcolor{black}{System Parameters}}
	\begin{tabular}{|l|l|}
		\hline
		\centering
		{Carrier frequency}  & {$750$MHz}\\
		\hline
		\centering
		{Bandwidth}  & {$1$MHz}\\
		\hline
		\centering
		Path loss at the reference distance of 1 meter  & $\rho_0=-30$dB \\
		\hline
		\centering
		Rician factor of the RIS assisted channels & $K_{AS}=K_{SU}= 3$dB \\ 
		\hline
		\centering
		Path-loss exponents of the RIS assisted channels   & $\alpha_{AS}=\alpha_{SU} = 2.2$    \\ 
		\hline
		\centering
		Maximum power budget  & $P_{\max} = 10$ W    \\ 
		\hline
		\centering
		Noise power at receivers  & $\sigma^2 = -90$ dBm    \\ 
		\hline
		\centering
		Initialized penalty factor for Algorithms 2  & $\xi={10^{ - 4}}$   \\ 
		\hline
		\centering
		Maximum number of iterations for Algorithm 1 and 2 & $L_{\max}=30$, $K_{\max}=20$ \\ 
		\hline  
		\centering
		Convergence accuracy   & ${{\varepsilon_0}}={{\epsilon_1}}={{\varepsilon_2}}={10^{ -3}}, {\epsilon_2}={10^{ - 7}}$  \\ 
		\hline
		\centering
		Search step size $s_t$  for $\epsilon$-constraint method & $\triangle\delta=0.1$  \\ 
		\hline
	\end{tabular}
	\centering
\end{table*}
	
To demonstrate the performance improvements introduced by deploying the STAR-RIS in a SWIPT system, two baselines are considered for comparison. 1) \textbf{Baseline scheme 1 (also referred to as reflecting-only RIS)}: In this case, a reflecting-only RIS is deployed at the $(15,6,0)$ meters to facilitate communication for all EUs and IUs via the passive reflective beamforming. 2) \textbf{Baseline scheme 2 (also referred to as conventional RIS)}: In this case, a reflecting-only RIS and a transmitting-only RIS are deployed adjacent to each other for communication. For fairness in comparison, each conventional RIS consists of $M/2$ elements, and the coefficient matrices are regarded as $\boldsymbol{\beta}^t=[\mathbf{1}_{1\times M/2},\mathbf{0}_{1\times M/2}]^T $ for transmitting-only RIS and $\boldsymbol{\beta}^r=[\mathbf{0}_{1\times M/2},\mathbf{1}_{1\times M/2}]^T$ for reflecting-only RIS. It is worth noting that the resulting optimization problems for baseline schemes can also be solved directly by applying \textbf{Algorithm 2}. It is important to emphasize that the following results shown (i.e., Figs. 3-9) are obtained by averaging over 50 channel realizations unless otherwise specified.
\begin{figure}[]
	\centering
	\setlength{\belowcaptionskip}{0cm}   
	\includegraphics[width=3.3in]{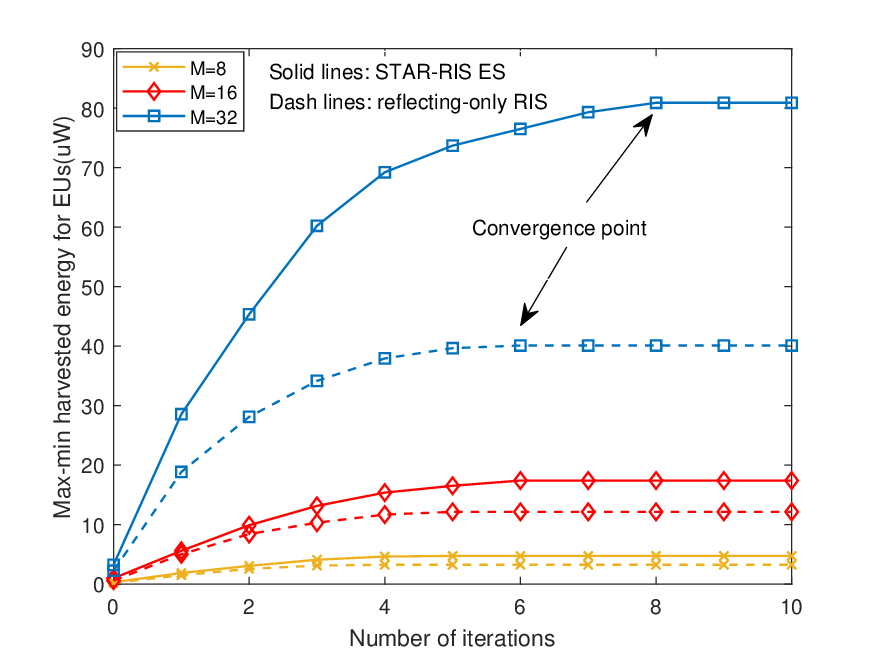}
	\caption{Convergence of \textbf{Algorithm 2}.}
	\label{convergence}
\end{figure}	
\subsection{Convergence of Algorithm 2}
In Fig. \ref{convergence}, we show the convergence behavior of the proposed \textbf{Algorithm 2} for STAR-RIS ES and reflecting-only RIS with different STAR-RIS elements $M$. Specially, we set $N=4$, $\rho_H=\rho_G=0.01$ and $\epsilon=0$. The results obtained for one random channel realization depict that the max-min harvested power of EUs increases quickly as the number of iterations increases, and finally converges to a value within $6$ iterations for $M=8$ and $M=16$. Even when $M=32$, our proposed algorithm can achieve convergence at the $8$-th iteration. However, compared to the reflecting-only RIS, the proposed algorithm for the STAR-RIS converges more slowly. This is expected since the computational complexity increases with more variables to be optimized for the STAR-RIS with larger-scale elements.
	
\subsection{System Performance Versus Number of AP Antennas}
In Fig. \ref{Rate-energy1}, we examine the achievable rate-energy region in relation to the number of AP antennas. We set $M = 8$, and $\rho_G = \rho_H = 0.01$. The results depict that the rate-energy region for all schemes expands with the number of antennas due to the active beamforming gain. In addition, regions obtained by the proposed scheme are larger than those obtained with reflecting-only RISs and conventional RISs because the former can take advantage of more DoFs. Further, regarding the performance for two protocols of STAR-RISs, ES is able to achieve both higher upper boundaries, i.e., $R_{\max}$ and $E_{\max}$, for IUs and EUs, respectively. However, TS can ensure better a performance balance for all users. This can be explained as follows. Compared to TS, ES accommodates all users to utilize the entire time resource for communication or charging, thus enabling a better upper bound on data rate$/$harvested power when only focusing on IUs or EUs. However, the time allocation for TS allows the AP and the STAR-RIS to serve users in only one region during each allocated time period, which reduces competition for communication resources between IUs and EUs as well as interference among IUs. As a result, the decline of harvested power is more moderate as the data rate increases, leading to a better balance between IUs and EUs. 
	
In Fig. \ref{Harvested power1}, we further explore the max-min harvested power for EUs versus the number of AP antennas. We set the max-min data rate of IUs to be $R_{\min}=4$ bit/s/hz, $M=16$, and $\rho_G=\rho_H=\rho=0.01$. As depicted in Fig. \ref{Harvested power1}, the max-min harvested power for all schemes increases with the number of AP antennas, and STAR-RISs outperform reflecting-RISs and conventional RISs. This is expected because, compared with reflecting-RIS, although adopting STAR-RIS leads to energy leakage or time loss for each user, the flexible deployment of STAR-RIS can provide better channel conditions, which can make up for the loss of communication resources. More importantly, the enhanced DoFs exploited by STAR-RISs can further boost desired signals and suppress unwanted ones, thereby achieving a significant performance improvement.
	
\begin{figure}[]
	\subfigure[Rate-energy region.]{\label{Rate-energy1}
		\includegraphics[width= 3.3in]{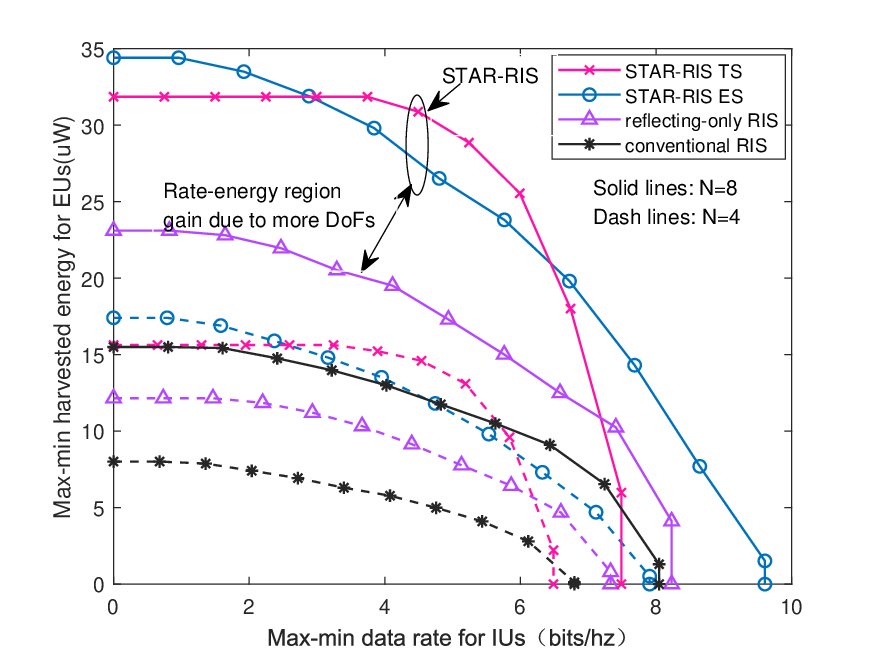}}
	\subfigure[Max-min harvested power.]{\label{Harvested power1}
		\includegraphics[width= 3.3in]{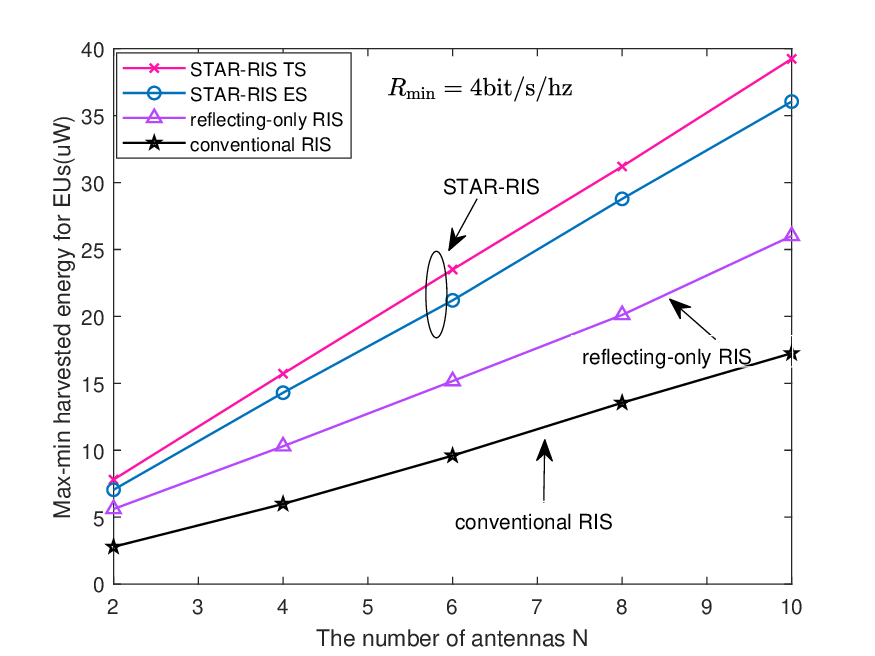}}
	\caption{System performance  versus number of AP antennas for $M=16$, $\rho_G=\rho_H=0.01$} 
\end{figure}
\label{}
	
\subsection{System Performance Versus Number of STAR-RIS Elements}
Fig. \ref{Rate-energy2} shows the achievable rate-energy region versus the number of STAR-RIS elements. We set $N=4$, and $\rho_G=\rho_H=0.01$. Firstly, it is seen that the rate-energy region grows for all schemes as $M$ increases, and the performance gap between the proposed design and the baseline schemes becomes more pronounced. This is because more elements lead to higher transmission/reflection beamforming gains and DoFs benefits. Secondly, the power gain for EUs is more significant than the rate gain for IUs. This can be explained by citing the following causes.  The calculation of the information rate requires a logarithmic operation, which weakens the gain brought by RISs to received SINR. Whereas, the calculation of harvested power does not involve logarithmic operations, thus resulting in the differences in growth rates. Similar reasons can be used to explain the higher performance gains for EUs than IUs via the introduction of STAR-RISs.
	
In Fig. \ref{Harvested power2}, we further investigate the max-min harvested power for EUs versus the number of STAR-RIS elements under the max-min data rate $R_{\min}=4$ bit/s/hz. We set $N=4$, and $\rho_G=\rho_H=0.01$. As can be observed, the max-min harvested power for all schemes increases with the STAR-RIS elements. Particularly, STAR-RISs rise noticeably faster than reflecting-RISs and conventional RISs. This is because the extra DoFs for STAR-RISs can extend the passive beamforming gains by more elements. Besides, the gap between TS and ES increases as $M$ increases. This is made possible by the fact that interference-free communication for TS can make up for the inefficient use of the communication time when the $M$ is large.

\begin{figure}[]
	\centering
	\subfigure[Rate-energy region.]{\label{Rate-energy2}
		\includegraphics[width= 3.3in]{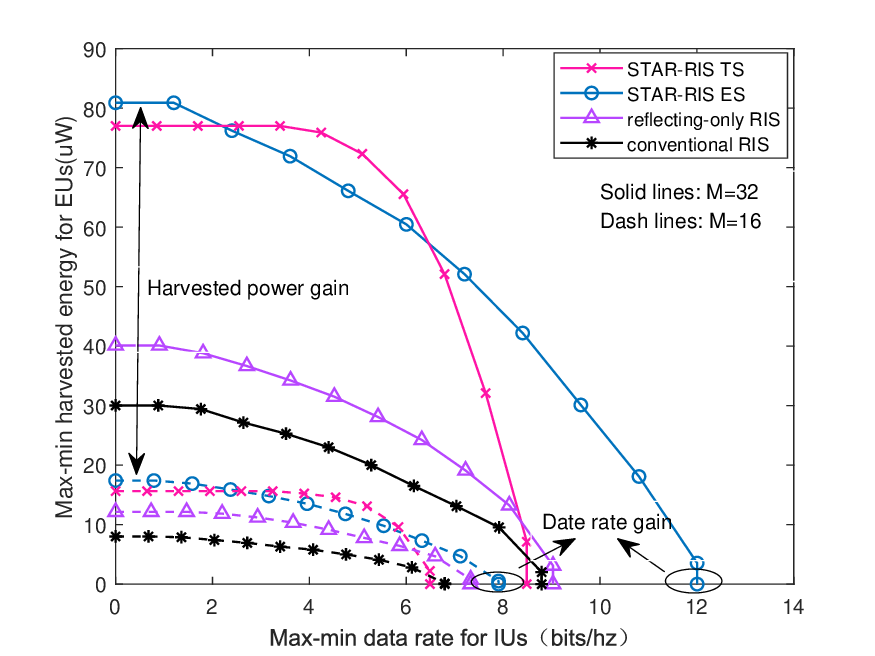}}
	\subfigure[Max-min harvested power.]{\label{Harvested power2}
		\includegraphics[width= 3.3in]{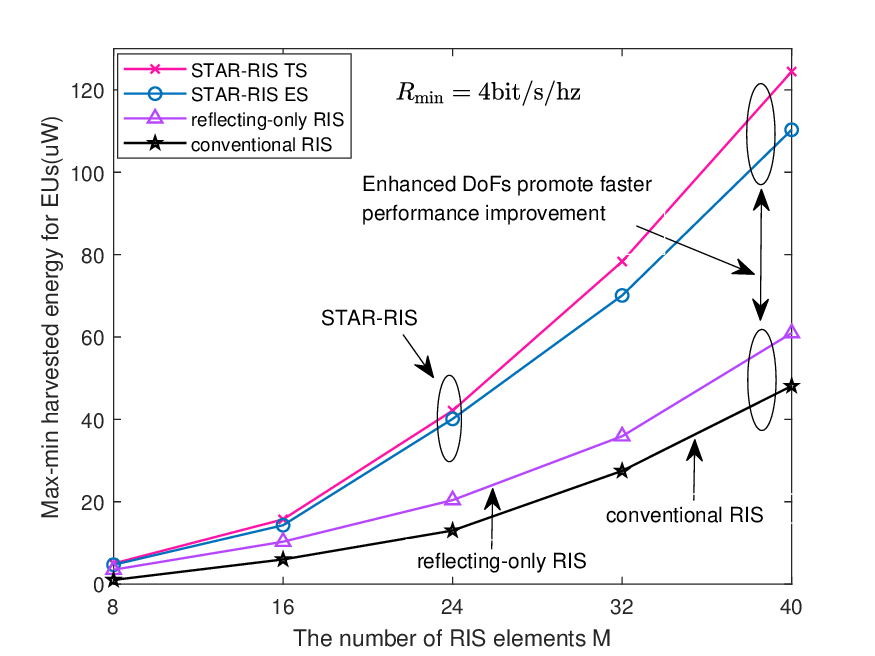}}
	\caption{System performance  versus number of STAR-RIS elements for $N=4$, $\rho_G=\rho_H=0.01$} 
\end{figure}
\label{}
\begin{figure}[]
	\centering
	\subfigure[Rate-energy region.]{\label{Rate-energy3}
		\includegraphics[width= 3.3in]{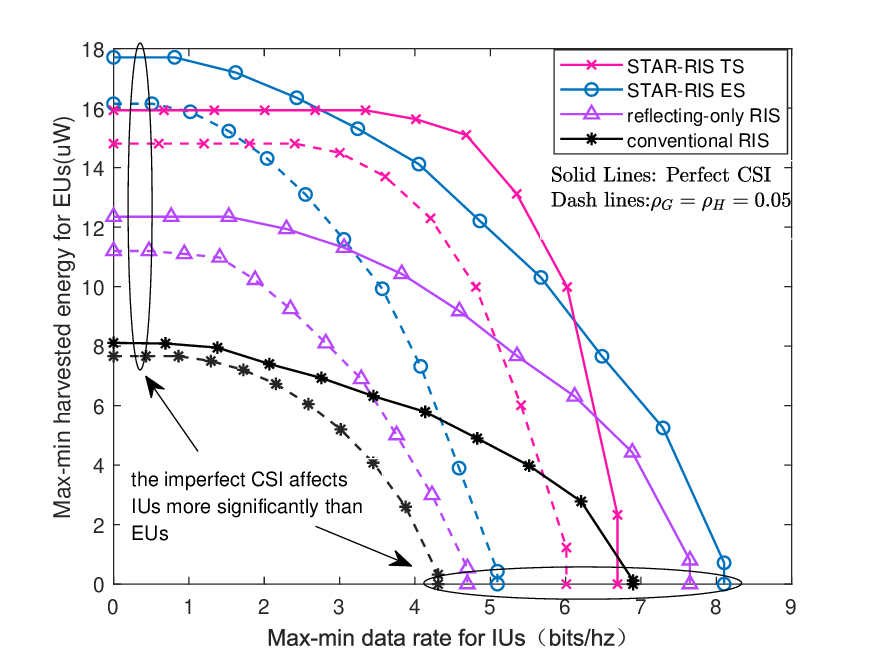}}
	\subfigure[Max-min harvested power.]{\label{Harvested power3}
		\includegraphics[width= 3.2in]{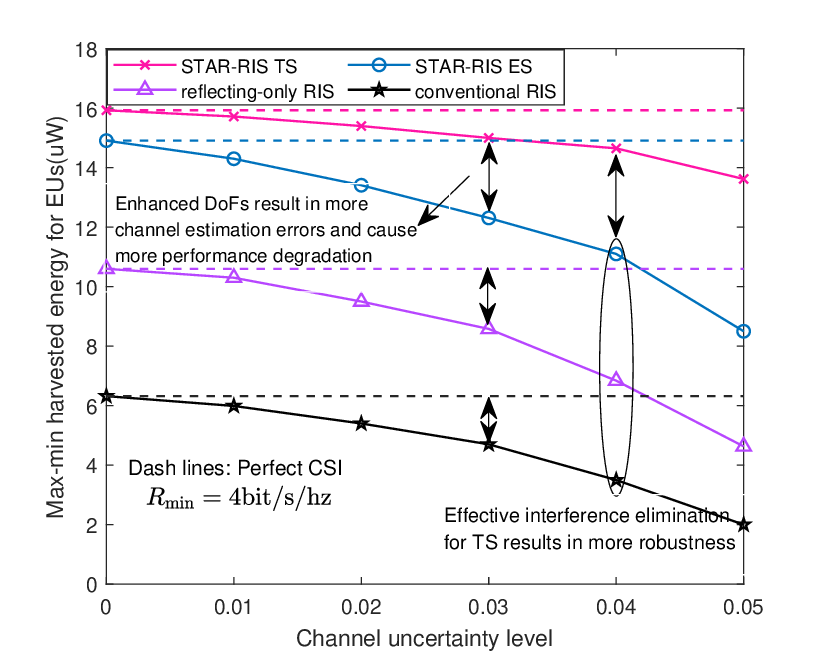}}
	\caption{System performance versus CSI uncertainty levels for $M=16$, $N=4$} 
\end{figure}
\label{CSI uncertainty level}
\subsection{System Performance Versus CSI Uncertainty Levels}
In Fig. \ref{Rate-energy3}, we study the achievable rate-energy region versus the cascaded channel uncertainty levels. We set $M = 16$, $N = 4$, and $\rho_G= \rho_H = \rho$ for the system. Simulation results depict that the performance region for all schemes decreases as $\rho$ increases, and the performance decline of our proposed scheme is more pronounced than the baseline schemes employing reflecting-only RISs and conventional RISs. This can be understood by pointing out that increasing design DoFs leads to a larger channel estimation error and reduces the robustness of user performance. In particular, the channel error has a greater impact on the IUs than the EUs for  STAR-RISs working in ES mode and conventional RISs. This is because, the uncertainty of the channel can not only diminish the desired signal received by IUs like EUs, but also enhance the interference from other IUs.
However, TS is able to minimize the interference between IUs, which leads to more robustness for IUs under imperfect CSI.
	
Fig. \ref{Harvested power3} further investigates the max-min harvested power versus the cascaded channel uncertainty levels under the max-min data rate $R_{\min}=4$ bit/s/hz. We set $M = 16$, $N = 4$, and $\rho_G = \rho_H = \rho$ for the system. As can be observed, the max-min harvested power for all schemes decreases more and more as $\rho$ increases. This is expected since the channel errors affect IUs more than EUs, and as $\rho$ increases to also satisfy $R_{\min}=4$ bit/s/hz, more communication resources need to be allocated to IUs, thus significantly weakening the performance of EUs. While TS can effectively eliminate interference with IUs, which results in more robustness. On the other hand, the performance degradation for ES is more pronounced than for conventional RISs. This is due to the fact that increased DoFs for STAR-RISs will lead to higher channel estimation errors.

\begin{figure}[]
	\centering
	\setlength{\belowcaptionskip}{0cm}   
	\includegraphics[width=3.3in]{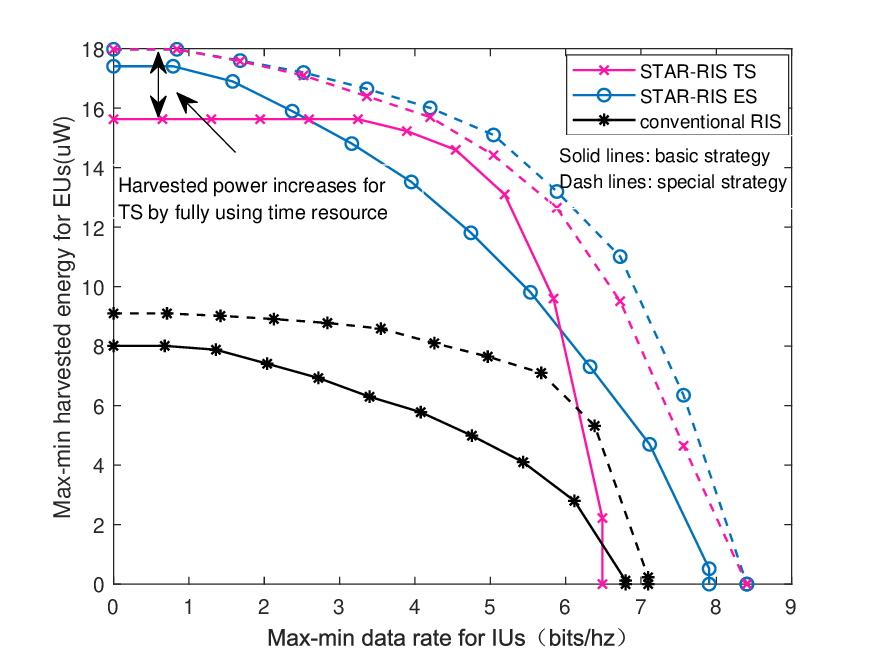}
	\caption{Rate-energy region versus STAR-RIS deployment strategies}
	\label{Deployment}
\end{figure}	
\subsection{Rate-energy Region Versus STAR-RIS Deployment Strategies}
In Fig. \ref{Deployment}, we investigate the impact of STAR-RISs deployment strategies on achievable rate-energy region. In addition to the deployment strategy of the basic scenario described above, here we consider a special deployment strategy of STAR-RISs that makes all IUs and EUs located in the T and R regions, respectively. For fairness of comparison, we set $M = 16$, $N = 4$ and $\rho_G = \rho_H = 0.01$. As can be seen, the special deployment strategy outperforms basic deployment strategy  for all schemes. Following are some explanations for the causes of this. Considering the STAR-RISs ES protocol and conventional RISs, the phase shift design of RISs for the basic deployment strategy needs to balance the requirements of both IU and EU in the same region, but also to suppress the interference from IU located in another region, which cuts the passive beamforming gains. But the special deployment strategy only considers competition between users located in the same region for resource allocation, which reduces the energy leakage for each IU or EU and enhances resource dedication. While for the STAR-RISs TS protocol, more time resources can be allocated for all IUs and EUs by employing the special deployment strategy, which further compensates for the lack of time available in the basic deployment strategy. Therefore, the special deployment strategy results in a better rate-energy region. However, the stringent requirements for user distribution entail higher implementation difficulties, especially as the number of users grows.
	
\subsection{Rate-energy Region Versus Filters at the HAP}
\begin{figure}[]
	\centering
	\setlength{\belowcaptionskip}{0cm}   
	\includegraphics[width=3.3in]{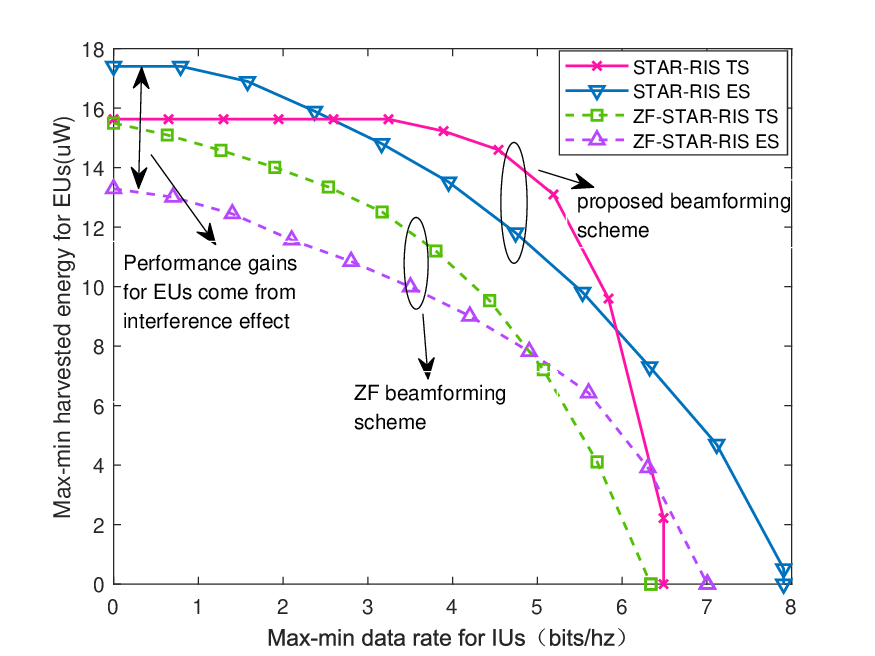}
	\caption{Rate-energy region versus filters at the HAP }
	\label{ZF}
\end{figure}	
In Fig. \ref{ZF}, we study the achievable rate-energy region versus different filters at the HAP. Here, the zero-forcing (ZF) beamforming scheme is introduced as a comparison for our proposed beamforming scheme under both the ES and TS protocols. In the basic scenario, we set $N=4, M=16$, and $\rho_G=\rho_H=0.01$. The Fig. 8 depicts that our proposed scheme outperforms the ZF beamforming scheme for both ES and TS. Moreover, the performance improvement for EUs is greater than that for IUs. The reasons can be explained as follows. When the number of antennas is sufficient, ZF beamforming can fully suppress the interference between users. However, due to stringent interference elimination requirements, ZF beamforming also inevitably reduces the flexibility of communication resource allocation, which results in poorer performance gains for IUs than our proposed beamforming scheme. While for EUs, unlike ZF's obsession with eliminating all interference, our proposed scheme fully exploits the potential of constructive interference and suppresses destructive interference for each EU. As a result, our scheme achieve a more significant performance gain in EUs than ZF beamforming. On the other hand, since the TS protocol has effectively attenuated inter-user interference, there is a smaller performance gap between our proposed scheme and ZF beamforming compared to ES.
	
\subsection{Rate-energy Region Versus  Number of Users}
Fig. \ref{Users} characterizes the achievable rate-energy region versus the number of users. Compared to the basic setup, two new scenarios are considered with 4 EUs, 2 IUs and 2 IUs, 4 EUs, respectively, where we set $M = 16$, $N = 4$ and $\rho_G = \rho_H = 0.01$. For ease of exploration, we assume all IUs and EUs are uniformly located in the T and R regions. As clearly shown in Fig. \ref{Users}, as the number of users increases, the obtained rate-energy region becomes smaller. This is a result of that when a new user accesses the system, in order to ensure fairness among users, the communication resources allocated to each user will decrease. Especially, increasing the number of IUs, the system performance degrades more significantly, as expected. This is because IUs are not only limited by their individual QoS requirements, but can also cause significant inter-user interference to other IUs. This fact also leads to TS outperforming ES when the number of IUs is large. Because TS  provides a good suppression on the inter-user interference, which can impair the impact of lost time resources and thus contribute to achieving higher performance gains.
\begin{figure}[]
	\centering
	\setlength{\belowcaptionskip}{0cm}   
	\includegraphics[width=3.3in]{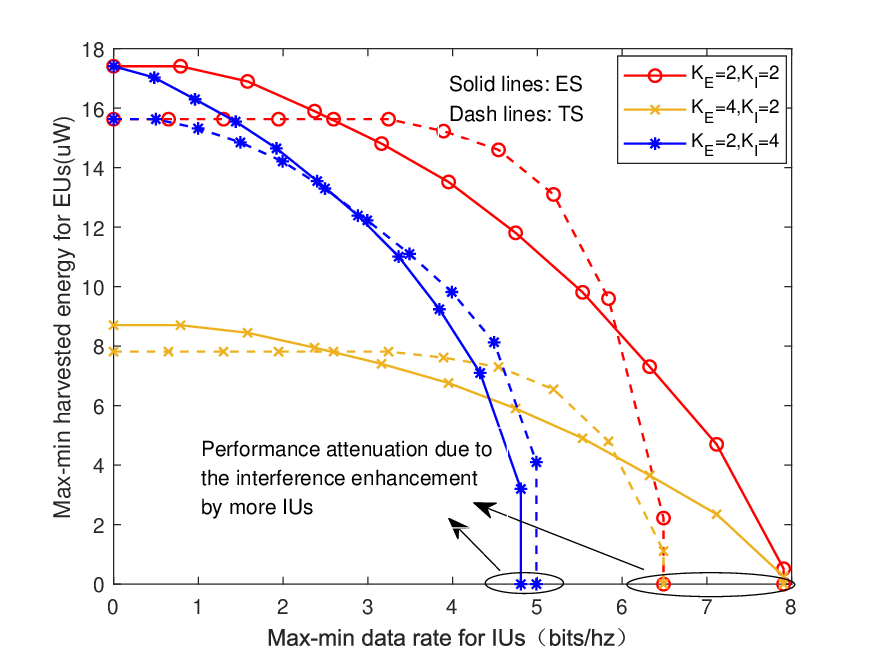}
	\caption{Rate-energy region versus number of users}
	\label{Users}
\end{figure}	
	
\section{Conclusions}
In this paper, the robust resource allocation design for STAR-RIS enabled SWIPT systems was investigated. Based on the assumption of imperfect CSI, two MOOP frameworks for different operating protocols were deployed to study the fundamental trade-off between the max-min data rate and harvested power. For each protocol, the MOOP was first converted into a SOOP by applying the $\epsilon$-constraint method, and then an AO-based algorithm was proposed to explore the robust resource allocation design. Numerical results unveiled that STAR-RISs can greatly outperform conventional reflecting-only RISs in terms of system performance, especially for EUs. Also, 
the ES and TS protocols of STAR-RISs can enhance the SWIPT system from different perspectives. Imagine that, combined with implementation difficulties, these insights will offer helpful suggestions for the resource allocation design of STAR-RIS aided SWIPT systems in practical scenarios.

\section*{Appendix~A: Proof of Theorem 1} \label{Appendix:B}
The relaxed problem in \eqref{P7} is jointly convex with respect to the optimization variables $\{\mathbf{W}_i,\mathbf{V}_j\}$ and satisfies Slater’s constraint qualification \cite{Rank_one_proof}. Therefore, strong duality of  holds. Here, 
we take the Lagrangian function of \eqref{P7} into consideration as $\textup{(39)}$, which is shown at the top of the next page,
\begin{figure*}
	\normalsize
	\begin{align}
		\nonumber
		L&=\eta-\lambda\left(\sum_{i\in\mathcal{K_E}}\mathrm{Tr}(\mathbf{W}_i)+\sum_{j\in\mathcal{K_E}}\mathrm{Tr}(\mathbf{V}_j)-P_{\max}\right)+\sum_{j\in\mathcal{K_E}}\mathrm{Tr}\left(\mathbf{C}_{1,j}\mathbf{S}_{1,j}\right)+\sum_{i\in\mathcal{K_I}}\mathrm{Tr}\left(\mathbf{C}_{2,i}\mathbf{S}_{2,i}\right) \\ 
		&+\sum_{i\in\mathcal{K_I}}\mathrm{Tr}\left(\mathbf{Y}_i\mathbf{W}_i\right)+\sum_{j\in\mathcal{K_E}}\mathrm{Tr}\left(\mathbf{Z}_j\mathbf{V}_j\right)+ \mathbf{\triangle}.
	\end{align}
	\hrulefill \vspace*{0pt}
\end{figure*}
where we denote $\mathbf{C}_{1,j}$ and $\mathbf{C}_{2,i}$ as the abbreviations for LMIs in constraints $\textup{(20)}$ and $\textup{(21)}$, respectively. And $\lambda \geq 0$, $\mathbf{S}_{1,j} \succeq 0$ and $\mathbf{Z}_j \succeq 0, \forall j \in \mathcal{K_E}$, $\mathbf{S}_{2,i} \succeq 0$ and $\mathbf{Y}_i \succeq 0, \forall i \in \mathcal{K_I}$ are the dual variables for constraints $\textup{(15c)}$, $\textup{(20)}$ and $\textup{(15e)}$, $\textup{(21)}$ and $\textup{(15d)}$, respectively. 
Besides, all terms that are unrelated to the proof make up the collection $\triangle$.
Now, we focus on those Karush-Kuhn-Tucker (KKT) conditions with respect to $\mathbf{W}_i$ and $\mathbf{V}_j$, thus obtained as follows:
\begin{subequations}
	\begin{gather}
		\lambda^*\geq 0, \mathbf{S}^{*}_{1,j} \succeq 0, \mathbf{S}^{*}_{2,i} \succeq 0, \mathbf{Y}^{*}_i, \mathbf{Z}^{*}_j \succeq 0,\\
		\mathbf{Y}^{*}_i\mathbf{W}^{*}_i=0, \mathbf{Z}^{*}_j\mathbf{V}^{*}_j=0, \\ \nonumber
		\mathbf{Y}^{*}_i\quad=\lambda^{*}\mathbf{I}_{N}-
		\big(\sum_{l=1}^{M}\left[\text{vec}(\mathbf{H}_i)^{H}\frac{\mathbf{C}^{*}_{2,i}}{\Gamma_i} \textup{vec}(\mathbf{H}_i)\right]_{a:b,c:d}\!\!\!\\ 
		-\sum_{k \neq i}\sum_{l=1}^{M}\left[\text{vec}(\mathbf{H}_i)^{H}\mathbf{C}^{*}_{2,k} \textup{vec}(\mathbf{H}_i)\right]_{a:b,c:d} \! \! \nonumber \\ 
		\quad\quad\quad\quad+\!\sum_{j\in\mathcal{K_E}}\!\!\sum_{l}^{M}\left[\text{vec}(\mathbf{G}_j)^{H}\mathbf{C}^{*}_{1,j} \textup{vec}(\mathbf{G}_j)\right]_{a:b,c:d}\big), \\
		\mathbf{Z}^{*}_j = \lambda^{*}\mathbf{I}_N-\sum_{j\in\mathcal{K_E}}\sum_{l}^{M}\left[\text{vec}(\mathbf{G}_j)^{H}\mathbf{C}^{*}_{1,j} \textup{vec}(\mathbf{G}_j)\right]_{a:b,c:d} ,
	\end{gather}
\end{subequations}
where $a=c=(l-1)N+1$ and $b=d=lN$. As can be observation from $\textup{(40b)}$, the rank of $\mathbf{W}^{*}_i$ is related to the rank of $\mathbf{Y}^{*}_i$. Considering a feasible scenario with $P_{\max}>0$ and $\Gamma_i>0$, $\lambda>0$ and $\mathbf{W}_i \neq \mathbf {0}$ must always hold. 
Then, it can be demonstrated that $\mathrm{Rank}\left(\mathbf{W}^{*}_i\right)=1$ by taking use of the findings in \cite{Rank_one_proof}. Similarly, applying the method of \cite{Rank_one_proof}, we can also prove $\sum_{j=1}^{K_E}\mathrm{Rank}\left(\mathbf{V}^{*}_j\right) \leq 1$. 
\balance
\bibliographystyle{IEEEtran}
\bibliography{mybib.bib}
\balance
\end{document}